\documentclass[prl,reprint,showpacs,superscriptaddress]{revtex4-1}
\usepackage[T1]{fontenc}
\usepackage{amsmath, amsthm, amssymb, mathtools, dsfont}
\usepackage{times}
\usepackage{graphicx}
\usepackage{dcolumn}
\usepackage{bm,bbm}
\usepackage[normalem]{ulem}
\usepackage{color}
\usepackage[usenames,dvipsnames]{xcolor}
\usepackage{esint}
\usepackage{paralist}

\usepackage{geometry}
\definecolor{myurlcolor}{rgb}{0,0,0.7}
\definecolor{myrefcolor}{rgb}{0.8,0,0}
\usepackage[unicode=true,pdfusetitle, bookmarks=true,bookmarksnumbered=false,bookmarksopen=false, breaklinks=false,pdfborder={0 0 0},backref=false,colorlinks=true, linkcolor=myrefcolor,citecolor=myurlcolor,urlcolor=myurlcolor]{hyperref}

\DeclareGraphicsExtensions{.png,.pdf,.eps}
\graphicspath{ {./figures/} }
\usepackage{epstopdf}

\makeatletter

 \theoremstyle{plain}
 \newtheorem{fact}{Fact}
 \theoremstyle{plain}
 \newtheorem{lem}{Lemma}
 \theoremstyle{plain}
 \newtheorem{thm}{Theorem}
 \theoremstyle{plain}
 \newtheorem{exa}{Example}
 \theoremstyle{plain}
 
 \theoremstyle{plain}
 
 \theoremstyle{remark}
 \newtheorem*{rem*}{Remark}
 \theoremstyle{plain}
  \newtheorem{rem}{Remark}

\newcommand{\ot}{\otimes}

\renewcommand{\exp}{\mathrm{exp}}
\DeclareMathOperator{\tr}{tr}

\renewcommand{\H}{\mathcal{H}}

\newcommand{\defeq}{\coloneqq}

\renewcommand{\C}{\mathbb{C}} 
\newcommand{\R}{\mathbb{R}} 
\newcommand{\M}{\mathbf{M}} 
\newcommand{\N}{\mathbf{N}} 
\newcommand{\T}{\mathbf{T}} 
\newcommand{\W}{\mathcal{W}} 
\newcommand{\Q}{\mathcal{Q}} 

\newcommand{\X}{\mathbf{X}} 

\renewcommand{\P}{\mathcal{P}} 
\newcommand{\PPP}{\mathbf{P}} 

\newcommand{\PP}{\mathbb{P}} 

\newcommand{\SP}{\mathrm{S}\mathbb{P}} 

\newcommand{\Herm}{\mathrm{Herm}} 
\newcommand{\I}{\mathbb{I}} 

\global\long\global\long\global\long\def\bra#1{\mbox{\ensuremath{\langle#1|}}}
\global\long\global\long\global\long\def\ket#1{\mbox{\ensuremath{|#1\rangle}}}
\global\long\global\long\global\long\def\bk#1#2{\mbox{\ensuremath{\ensuremath{\langle#1|#2\rangle}}}}
\global\long\global\long\global\long\def\kb#1#2{\mbox{\ensuremath{\ensuremath{\ensuremath{|#1\rangle\!\langle#2|}}}}}

\global\long\global\long\global\long\def\SET#1#2{\mbox{\ensuremath{\ensuremath{\left\lbrace\left. #1\ \right|\ #2 \right\rbrace }}}}

\makeatother

\begin{document}

\author{Micha{\l}  Oszmaniec}
\affiliation{ICFO-Institut de Ciencies Fotoniques, The Barcelona Institute of Science and Technology, 08860 Castelldefels (Barcelona), Spain}
\email{michal.oszmaniec@icfo.es} 

\author{Leonardo Guerini}
\affiliation{ICFO-Institut de Ciencies Fotoniques, The Barcelona Institute of Science and Technology, 08860 Castelldefels (Barcelona), Spain}
\affiliation{Departamento de Matem\'atica, Universidade Federal de Minas Gerais, Caixa
Postal 702, 31270-901, Belo Horizonte, MG, Brazil}

\author {Peter Wittek}
\affiliation{ICFO-Institut de Ciencies Fotoniques, The Barcelona Institute of Science and Technology, 08860 Castelldefels (Barcelona), Spain}
\affiliation {University of Bor{\aa}s, 50190 Bor{\aa}s, Sweden}

\author{Antonio Ac\'\i n}
\affiliation{ICFO-Institut de Ciencies Fotoniques, The Barcelona Institute of Science and Technology, 08860 Castelldefels (Barcelona), Spain}
\affiliation{ICREA-Instituci\'o Catalana de Recerca i Estudis Avan\c cats, Lluis Companys 23, 08010 Barcelona, Spain}

\title{Simulating positive-operator-valued measures with projective measurements}
\begin{abstract}
Standard projective measurements represent a subset of all possible-measurements in quantum physics, defined by positive-operator-valued measures. We study what quantum measurements are projective simulable, that is, can be simulated 
by using projective measurements and classical randomness. We first prove that every measurement on a given quantum system can be realised by classical processing of projective measurements on the system plus an ancilla of the same dimension. Then, given a general measurement in dimension two or three, we show that deciding whether it is projective-simulable can be solved by means of semi-definite programming. We also establish conditions for the simulation of measurements using projective ones valid for any dimension. As an application of our formalism, we improve the range of visibilities for which two-qubit Werner states do not violate any Bell inequality 
for all measurements. From an implementation point of view, our work provides bounds on the amount of noise a measurement tolerates before losing any advantage over projective ones.
\end{abstract}
\maketitle

According to the postulates of quantum mechanics, the statistics of a measurement of a quantum observable on a given quantum system is modelled by a set of orthogonal projectors acting on the Hilbert space of the system. For this reason, projective measurements (PMs) appear so
commonly in the context of quantum technologies, quantum foundations, and quantum information theory. It is known, however, that in quantum theory there exist more general measurements, corresponding to the so-called positive operator-valued measures (POVMs) defined by a set of positive operators summing up to the identity operator~\cite{Peres2006}. Being more general, POVMs outperform projective measurements for many tasks in quantum information theory, including quantum tomography~\cite{Renes2003}, unambiguous discrimination of quantum states~\cite{Bergou2010}, state estimation~\cite{Derka1998}, quantum cryptography~\cite{Bennet2002,Renes04Crypt,Nielsen2010}, information acquisition from a quantum source~\cite{Jozsa2003}, Bell 
inequalities~\cite{Gisin96,Vertesi2010} or device-independent quantum information protocols~\cite{Acin2016,Gomez2016}. 

Despite all these results, not much is known about the relation between general and projective measurements. In particular, given an arbitrary POVM, it is unknown whether it offers any advantage over projective ones or, on the contrary, can be replaced by projective measurements , which are easier to implement experimentally. The main objective of our work is to start the study of these questions. In particular, our goal is to characterise the set of PM-simulable measurements, i.e., those measurements that can be realised by performing projective measurements with the help of classical randomness (probabilistic mixing) and classical post-processing (see Figure~\ref{GenQuest}). We start by providing a formal definition of the set of PM-simulable measurements. We provide an alternative operational interpretation of this set by a generalisation of Naimark's theorem: to implement an arbitrary POVM on a Hilbert space $\H$ it suffices to have access to PM-simulable measurements on the extended system $\H\ot\H'$, where $\dim(\H')=\dim(\H)$.
We then move to the study of those POVMs that can be simulated by projective measurements performed solely on $\H$. 
We solve this problem completely for qubits and qutrits by giving a characterisation of the projective-simulable set
in terms of semi-definite programs (SDPs)~\cite{sdp} and providing explicit simulation strategies \footnote{All algorithms were implemented with PICOS -- Python Interface for Conic Optimization Software (available at \url{http://picos.zib.de/}) and we made the computational details available at \url{https://github.com/peterwittek/ipython-notebooks/blob/master/Simulating_POVMs.ipynb}.}. We also provide
a partial characterisation of PM-simulable measurements for arbitrary dimension. Finally, we illustrate the usefulness of our approach in the context of Bell inequalities. Based on ideas of the recent works~\cite{Dani2015,Hirsch2015}, we use our results to extend the range of visibilities for which qubit Werner states~\cite{Werner1989} have a local hidden-variable model for all measurements, improving the previous bound derived by Barrett ~\cite{Barrett2002}.

Apart from the theoretical interest, our results also have implications from an experimental perspective. In fact, while general measurements are more powerful than projective, their implementation is also more demanding. To implement a POVM on a Hilbert space $\mathcal{H}$, one usually has to control and probe not only the quantum system in question but also additional degrees of freedom \cite{Exp2009,FoundChile2011}. For instance, in an implementation with polarised photons, a projective measurements consists of a simple polarised beam-splitter, while a POVM requires coupling polarisation to other degrees of freedom \cite{Ota2012}, such as orbital angular momentum or spatial modes. It is then a usual situation than POVMs are noisier than projective measurements, which are almost noise-free. Within this context, our results provide bounds on the amount of noise a POVM tolerates before losing any advantage with respect to projective ones.

\begin{figure}[!h]
\centering
\includegraphics[width=0.6\columnwidth]{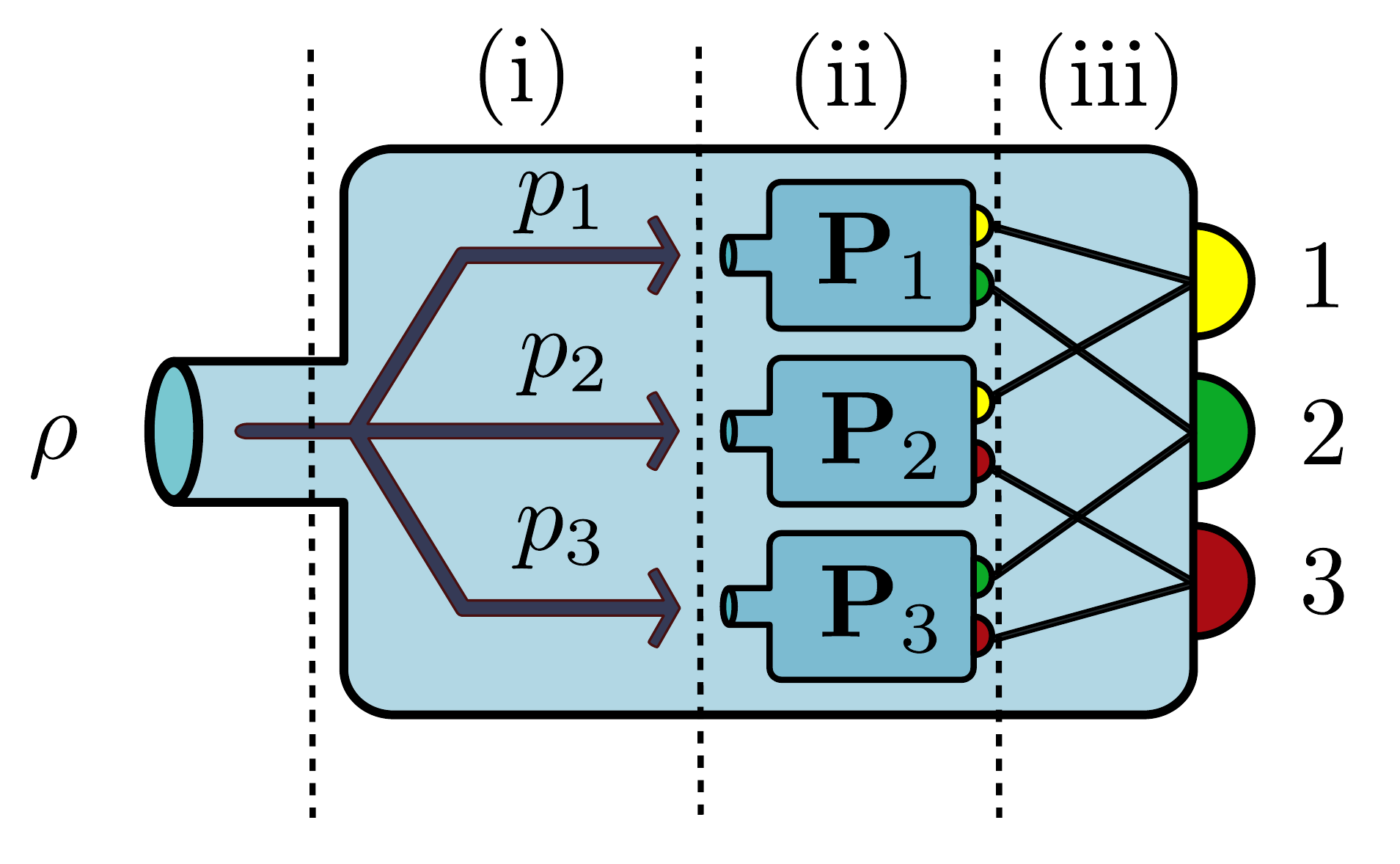}
\caption{We investigate which generalised measurements 
	 on a finite-dimensional Hilbert space $\H$ can be simulated if one has access only to (i) classical randomness, (ii) projective measurements on $\H$, and (iii) classical post-processing. The figure represents a three-outcome POVM which is simulated by combining three dichotomic projective measurements.}
\label{GenQuest}
\end{figure}
 
\paragraph*{Preliminaries---}

We start by introducing the basic concepts and notations used in what follows. 
By $\mathcal{H}$ we denote a $d$-dimensional Hilbert space ($\H\cong\C^{d}$) and by $\Herm(\H)$ the set of Hermitian operators on this space.
A POVM on $\mathcal{H}$ with $n$ outcomes is a vector $\M=\left(M_{1},\ldots,M_{n}\right)\in\Herm(\H)^{\times n}$ of non-negative operators satisfying $\sum_{i=1}^{n}M_{i}=\I$, where $\I$
is the identity operator on $\mathcal{H}$. 
The operators $M_{i}$ are called the effects of $\M$.
 We will denote the set of $n$-outcome POVMs on $\H$ by $\P\left(\H,n\right)$, or simply $\P\left(d,n\right)$ if only its dimension is relevant.
According to Born's rule, when a POVM $\M$ is measured on the quantum state $\rho$ the probability of obtaining the outcome $i$ is given by $p_i\left(\rho\right)=\mathrm{tr}\left(M_{i}\rho\right)$, $i=1,\ldots,n$. Given two POVMs $\M,\N\in\P\left(d,n\right)$, their convex combination $p\M+(1-p)\N$ is the POVM with $i$-th effect given by $\left[p\M+\left(1-p\right)\N\right]_i \defeq pM_{i}+(1-p)N_{i}$, and therefore $\P\left(d,n\right)$ is a convex set (see \cite{Sentis2013} for an efficient algorithm for decomposing an arbitrary POVM onto extremal POVMs). A projective (von Neumann) measurement is a POVM
whose effects are orthogonal 
projectors.
We denote the set of projective POVMs embedded in the space of $n$-output measurements $\P\left(d,n\right)$ by $\PP\left(d,n\right)$.
Notice that some of the outputs can have null effects and that effects are not required to be rank-one.
It is also useful to define $\P\left(d,n;m\right)$, the set of $m$-output measurements ($m\leq n$) considered as a subset of $\P(d,n)$.

Quantum measurements 
can be manipulated classically combining two different techniques~\cite{Buscemi2005,Haapasalo2012}: (i) randomisation (mixing), and (ii) post-processing \footnote{See also~\cite{Buscemi2005} for the systematic treatment of {\emph{quantum}} pre- and post-processing of POVMs.}.
Classical randomisation of the collection of POVMs $\left\lbrace \N_k \right\rbrace$ consists of choosing with probability $p_k$ the measurement $\N_k$ to be performed. This procedure yields a convex combination of POVMs $\left\lbrace \N_k \right\rbrace$,
$\N=\sum_k p_k  \N_k$. Classical post-processing of a POVM $\N$ is a strategy in which upon  obtaining an output $j$, one produces the final output $i$ with probability $q(i|j)$.
For a given post-processing strategy $q(i|j)$ this procedure gives a POVM $\Q \left(\N\right)$ with effects given by $\left[\Q\left(\N\right)\right]_i \defeq \sum_{j=1}^{n} q(i|j) N_j$. Notice that relabelling and coarse-graining are particular cases of post-processing.
We say that a POVM $\M\in\P(d,n)$ is PM-simulable if and only if it can be realised as classical randomisation followed by classical post-processing of some PMs $\textbf{P}_i \in \PP(d,n)$. We denote the class of PM-simulable $n$-outcome POVMs on $\H$ by $\SP(d,n)$.

We now introduce a quantitative measure of the non-projective character of POVMs. To do that, we estimate the amount of noise a measurement tolerates before it becomes PM-simulable. While several noise models are possible, here we take the simplest one defined by the depolarising map $\Phi_t(X)=tX+(1-t)\frac{\tr(X)}{d}\I$, 
which is also natural from an experimental point of view. When applied to a POVM $\M$, one gets ${
\left[\Phi_t\left(\M\right)\right]_i \defeq t M_i + (1-t)\frac{\tr(M_i)}{d} \I}$.
It is easy to see that $\Phi_1 (\M)=\M$ and that $\Phi_0(\M)$ is a trivial POVM $(\tr(M_1)\I/d,\ldots, \tr(M_n) \I/d)$ that can be realised by mixing the deterministic (projective) measurements $(\I,0,\ldots,0),\ldots, (0,\ldots,0,\I)$. Therefore, given $\M$, we are motivated to define the quantity
\begin{equation}\label{eq:visibil}
t(\M)\defeq \mathrm{max}\SET{t}{\Phi_t (\M) \in \SP(d,n)} \ ,
\end{equation}
that is, the maximal visibility ensuring that $\M$ is PM-simulable.
Due to the self-duality of the depolarising map $\Phi_t$, for every output $i$ and every state $\rho$ one has ${\mathrm{tr}\left(\rho \left[\Phi_t(M_i)\right] \right)=\mathrm{tr}\left([\Phi_t(\rho)] M_{i}\right)}$. In other words, the implementation of a POVM $\M$ on noisy states $\Phi_t(\rho)$ is statistically equivalent to measuring unperturbed states $\rho$ with the depolarised POVM $\Phi_t(\M)$. 

Finally, we denote by $t(d)$ the maximal visibility $t$ that ensures that $\Phi_t (\M)$ is projective-simulable for any $\M$ on $\C^d$,
\begin{equation}\label{eq:visibil}
t(d)\defeq \min\SET{t(\M)}{\M\in\P(d,d^2) }\ .
\end{equation}
In the above equation, the minimisation is over the set of $d^2$-output measurements on $\H$.
This is justified by the observation that the function $\M\mapsto t(\M)$ is concave and thus the minimum is obtained for some extremal POVM. Moreover, extremal POVMs on $\C^d$ have at most $d^2$ outcomes~\cite{DAriano2005,Chiribella2007}. Intuitively, the measurement $\M^*$ defined by the solution of the optimisation problem \eqref{eq:visibil} is the hardest to simulate with projective measurements. Figure~\ref{Geometry} provides a schematic representation of all the previous concepts.

\begin{figure}[!h]
\centering
\includegraphics[width=0.7\columnwidth]{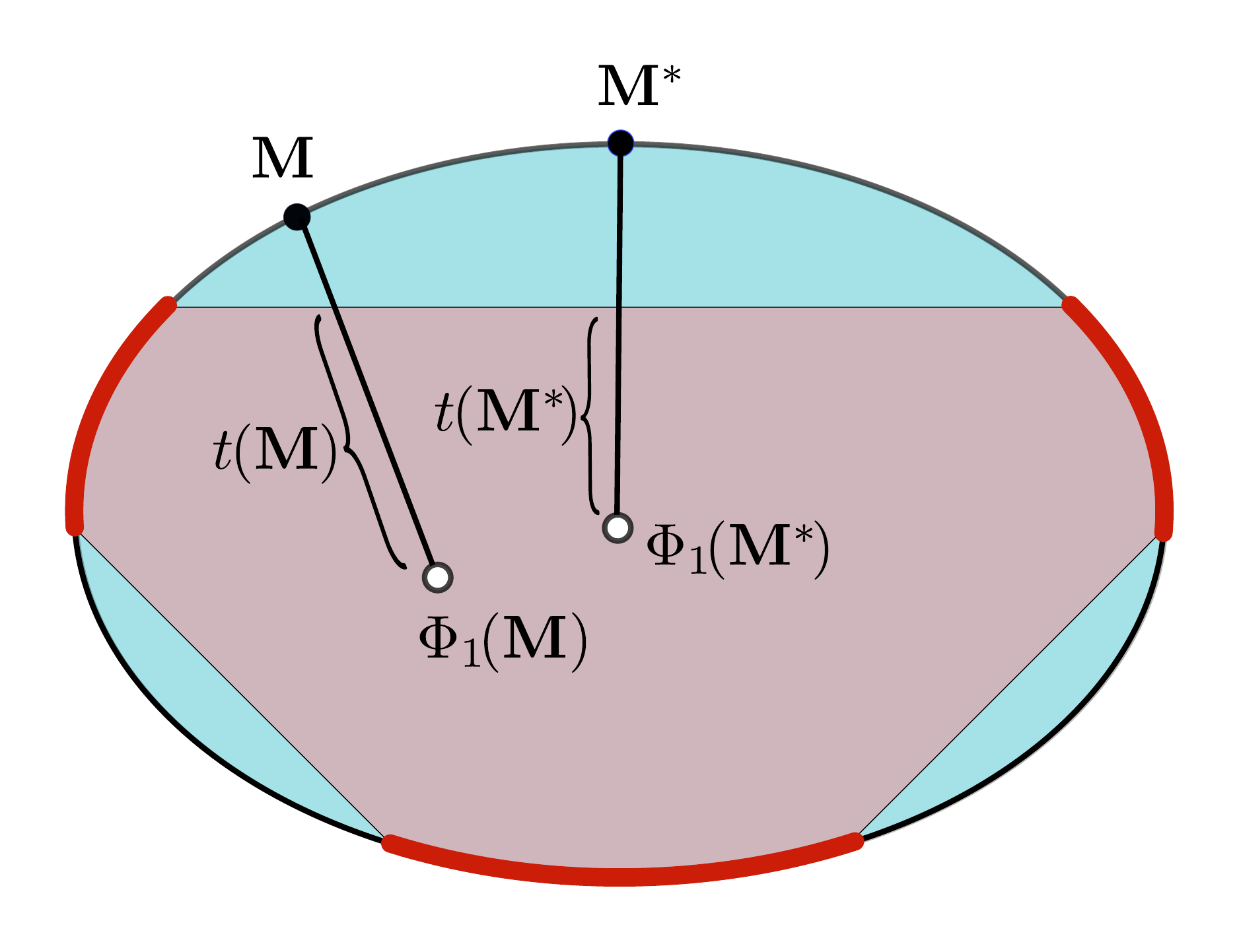}
\caption{(Color online) A schematic representation of the geometry of POVMs for given dimension. Solid red lines represent extremal projective measurements. Their convex hull (shaded in red) represents the set of PM-simulable measurements. For a non-projective POVM $\M$ the quantity $t(\M)$ denotes the critical visibility for which it becomes PM-simulable. The POVM most robust to noise is denoted by  $\M^\ast$ and is characterised by the smallest critical visibility  $t(\M^{\ast})$. }
\label{Geometry}
\end{figure}

\paragraph*{General results on projective simulability---}

Having introduced the main concepts, we start by showing that ultimately post-processing plays no role regarding projective simulability and the set of PM-simulable measurements can be obtained by just randomisation of PMs.

\begin{fact}[Operational characterisation of the convex hull of projective measurements]\label{convHullOp}

Let $\PP(d,n)$ and $\SP(d,n)$ be respectively the set of projective and PM-simulable measurements. Then, 
\begin{equation}\label{eq:operVSmat}
\SP(d,n)=\PP(d,n)^{\mathrm{conv}}\ ,
\end{equation}
where $\PP(d,n)^{\mathrm{conv}}$ is the convex hull of the set of PMs, consisting of POVMs of the form $\sum_i p_i \P_i$, where $\P_i\in\PP(d,n)$, $\ p_i\geq 0$ and $\sum_i p_i =1$.
\end{fact}

\begin{proof}
The inclusion $\PP(d,n)^{\mathrm{conv}} \subset \SP(d,n)$ follows from the fact that classical randomisation of $\PP(n,d)$ gives exactly $\PP(d,n)^{\mathrm{conv}}$.
The converse inclusion $\SP(d,n)\subset \PP(d,n)^{\mathrm{conv}}$ comes from the observation that the set $\PP(d,n)^{\mathrm{conv}}$ is preserved under classical post-processing. This is because extremal post-processing transformations are precisely relabellings and coarse-grainings that preserve the set of projective POVMs.
\end{proof}

Since not all extremal POVMs are projective~\cite{DAriano2005,Chiribella2007}, it follows from the above that for any $d$ not all POVMs on $\H$ are PM-simulable, as expected. An exemplary non-projective extremal POVM for $d=2$ is a tetrahedral measurement $\M^{\mathrm{tetra}}$.
This is a four-outcome POVM whose effects are given by $
M^{\mathrm{tetra}}_i =\frac{1}{4}\left(\I+\vec{n}_i \cdot \vec{\sigma}\right)$, 
where the unit vectors $\vec{n}_i$ are vertices of a regular tetrahedron inscribed in the Bloch sphere and $\vec \sigma$ is the vector of Pauli matrices.

The set of PM-simulable measurements is also interesting from the perspective of the following generalisation of Naimark's theorem for POVMs on $\H$.

\begin{thm}[Generalised Naimark's dilatation theorem]\label{thm:Naimark}
 Let $\SP(\H\ot\H', nd)$ be the set of projective-simulable, $nd$-outcome POVMs on $\H\ot\H'$, where $\dim(\H)=\dim(\H')=d$. Let $\M\in\P(\H,n)$ be an arbitrary $n$-outcome POVM on $\H$ and let $\kb{\phi}{\phi}$ be a fixed pure state on $\H'$. There exist a projective-simulable POVM $\N\in\SP(\H\ot\H', nd)$ such that 
\begin{equation}\label{eq:simulationTEH}
\tr(\rho M_i)= \tr(\rho\ot \kb{\phi}{\phi} N_i) 
\end{equation}
for $i=1,\ldots,n$, and all states $\rho$ on $\H$.
Moreover, $d$ is the minimum possible dimension for ancilla systems with this property.
\end{thm}

The above theorem follows from results on extremality of POVMs~\cite{DAriano2005} and the standard Naimark dilation theorem~\cite{Peres2006}, and can be found in Supplemental Material  ~\cite{supp}.
In previous studies, the dimension of the auxiliary system needed to implement a given POVM with a PM (without randomisation and post-processing) scaled linearly with the number of outputs~\cite{Peres2006,Jozsa2003,covORTOmes}, while here it remains constant. The mentioned optimality of dimension $d$ establishes a minimum dimension cost for the auxiliary system, even when one is allowed to use
classical randomisation and post-processing.

Our second general result valid for arbitrary dimension is a lower bound on the visibility needed for general projective simulability, $t(d)$. Explicit computation of this quantity is a difficult problem due to the bilevel nature of the optimisation problem \eqref{eq:visibil}. The computation of $t(\M)$ for any given $\M\in\P(d,n)$ provides only an upper bound to $t(d)$. 
\begin{lem}[Lower bound for the critical visibility for arbitrary dimension]\label{lem"lowerbound}
For any dimension $d$, we have $t(d)\geq 1/d$. 
\end{lem}
\begin{proof}[Sketch of the proof]
The idea is to first realise that for any measurement $\M$ and any post-processing $\Q$ we have $t(\Q(\M))\geq t(\M)$.
Then, it is enough to provide a projective simulation for the noisy rank-one POVMs, as every POVM consisting of operators of rank larger than one can be realised as coarse-graining of POVMs with rank-one effects. 
In Supplemental Material ~\cite{supp} we present an explicit simulation of rank-one POVMs using PMs for a visibility $t=1/d$.
\end{proof}


\paragraph*{Projective simulability for qubit and qutrit measurements---}


For $d=2$, PMs have at most 2 outcomes and we can reduce our problem to the study of 2-outcome POVMs.

\begin{lem}[Projective-simulable measurements for qubits]
\label{lem:main result}
For $d=2$, projective measurements can
simulate arbitrary two-outcome measurements.
In other words we have
\begin{equation} \label{eq:crucialfact1}
\SP\left(2,n\right) = \P\left(2,n;2\right)^{\mathrm{conv}} \ . 
\end{equation}
\end{lem}
\begin{proof}
Since projective measurements are a special case of 2-outcome measurements, we have $\mathbb{P}\left(2,n\right) \subset \P\left(2,n;2\right)$, and thus Lemma~\ref{convHullOp} implies that $\SP\left(2,n\right) \subset \P\left(2,n;2\right)^{\mathrm{conv}}$.
On the other hand, we know that every 2-output measurement is PM-simulable \cite{Masanes2005,Fritz2010}, $\P\left(2,n;2\right)\subset \SP\left(2,n\right)$.
Since $\SP\left(2,n\right)$ is a convex set, we have $\P\left(2,n;2\right)^{\mathrm{conv}}\subset \SP\left(2,n\right)$, which completes the proof.
\end{proof}

Characterisation \eqref{eq:crucialfact1} of PM-simulable measurements for qubits allows for efficiently describing this set in terms of an SDP  programme (see Supplemental Material~\cite{supp} for details). This SDP can be further modified to calculate the critical visibility $t(\M)$ for a fixed 4-outcome qubit POVM. The application of this SDP to a tetrahedral measurement $\M^{\mathrm{tetra}}$ yields the value $t_{\mathrm{tetra}}=\sqrt{2/3}\approx 0.8165$.
In Supplemental Material \cite{supp} we provide the analytical strategy for simulating $\M^{\mathrm{tetra}}$ via PMs for the visibility $t=\sqrt{2/3}$.

It is natural to suspect that for arbitrary $d$ the generalisation of Lemma~\ref{lem:main result} holds, in the sense that $\SP(d,n)=\P(d,n;d)^{\mathrm{conv}}$. However, this is already false for $d=3$. A counterexample is a modified trine measurement $\M\defeq\left(T_1, T_2, T_3 +\kb{2}{2}\right)\in\P(3,3)$, where $T_i = \frac23\kb{\psi_i}{\psi_i}$, $\ket{\psi_i}=\cos(\pi i /3 ) \ket{0}+ \sin(\pi i /3 ) \ket{1}$, are the effects of the trine qubit measurement $\M^{\mathrm{trine}}$~\cite{Jozsa2003}. The extremality of $\M$ follows from the the extremality of $\M^{\mathrm{trine}}$, since a decomposition of the former projected onto the subspace spanned by $\{\ket{0}, \ket{1}\}$ would provide a decomposition for the latter. Fortunately, for qutrits we still have a convenient description of the PM-simulable measurements.

\begin{lem}[Projective-simulable measurements for qutrits]
\label{lem:mainresult2}
For $d=3$, projective measurements can simulate 
2-outcome measurements and 3-outcome measurements with trace-1 effects. That is,
\begin{equation}\label{eq:crucialfact2}
\SP\left(3,n\right) = \left(\P\left(3,n;2\right) \cup \P_1\left(3,n;3\right)\right)^{\mathrm{conv}}\ \ , 
\end{equation}
where $\P_1\left(3,n;3\right) \subset \P\left(3,n\right)$ denotes the set of 3-output measurements on $\C^3$ with effects having unit trace.
\end{lem}
\begin{proof}[Proof sketch]
The main part of the proof is to show that the convex hull of qutrit rank-1 PMs simulate arbitrary 3-outcome POVMs having trace-1 effects (the intuition for that comes from the fact that effects of rank-1 PMs satisfy $\tr(M_i)=1$ and this property is preserved under
convex combinations).

We provide the details of the proof in Supplemental Material~\cite{supp}.
\end{proof}

Similarly to the case of qubits, the above characterisation of PM-simulable qutrit POVMs
reduces deciding whether a measurement $\M$ is projective-simulable to the computation of $t(\M)$ via an SDP. 
In Supplemental Material~\cite{supp} we give the explicit SDP programs together with the algorithm (based on the ``method of perturbations''~\cite{Chiribella2004}) extracting a specific projective-decomposition from their solution. 

Again, the generalisation of Lemma~\ref{lem:mainresult2} does not hold, in the sense that for $d > 3$ there are extremal trace-one $d$-outcome non-projective POVMs.
Indeed, consider the trace-1 POVM $\T = \M_{12} + \M_{34} \in \P(4,4)$ defined by the sum of two copies of $\M^{\mathrm{tetra}}$ supported in orthogonal two-dimensional subspaces of $\C^4$.
Again, $\T$ is extremal since the projections of its decomposition would imply a decomposition for $\M^{\mathrm{tetra}}$.

We conclude this part by studying the critical visibility for the case of qubits, $t(2)$. Clearly, $t(2)\leq t_{\mathrm{tetra}}=\sqrt{2/3}$. In what follows, we provide a lower bound to $t(2)$ that is very close to this value. To derive lower bounds on $t(2)$, we approximate the set of 4-outcome POVMs $\P(2,4)$ \emph{from the outside} by a polytope of ``quasi-POVMs'' $\Delta\subset\Herm(\C^2)^{\times 4}$. As by construction the set of 4-outcome qubit POVMs is strictly included in $\Delta$, the solution $t_\Delta$ to the following optimisation problem
\begin{equation}\label{eq:approxMin}
t_\Delta\defeq\min \SET{t(\M)}{\M\in\Delta}\ ,
\end{equation}
where $t(\cdot)$ is the function \eqref{eq:visibil} formally extended to $\Herm(\C^2)^{\times 4}$, provides a lower bound to $t(2)$. Since $t(\cdot)$ is a concave function, the minimum in \eqref{eq:approxMin} is attained for at least one of the vertices of $\Delta$. The advantage is now that $\Delta$ is a polytope and therefore has a finite number of vertices. Thus, to solve the optimisation problem \eqref{eq:approxMin} it suffices to compute the value $t(\mathcal{V})$ for each vertex $\mathcal{V}\in\Delta$, and take $t_\Delta=\min_\mathcal{V} t(\mathcal{V})$. As above, the value $t(\mathcal{V})$ can be computed efficiently in terms of an SDP program (see the discussion below Lemma~\ref{lem:main result}). The closer the outer polytope to the set of POVMs, the tighter the lower bound to $t(2)$.

Let us describe briefly our construction of the polytope of ``quasi-POVMs'' containing the set of general measurements. By definition, $(M_1,\ldots,M_4) \in\P(2,4)$ if and only if $\sum_{i=1}^4 M_i =\I$ and $M_i\geq 0$.
Recall that the positivity of effect $M_i$ is equivalent to $\tr(M_i \kb{\psi}{\psi})\geq 0$ for all pure states $\kb{\psi}{\psi}$. We define a polytope $\Delta$ by keeping the normalisation condition $\sum_{i=1}^4 M_i =\I$ and relaxing the positivity of effects of $\M$ by setting  $\tr\left(M_{i}\kb{\psi_j}{\psi_j}\right)\geq 0$, where $\ket{\psi_j}$ are suitably chosen pure states. In general, the more pure states are used, the better $\Delta$ approximates the set $\P(2,4)$. Also, to ensure that $\Delta$ is compact and thus is a polytope we have to assume that projectors $\kb{\psi_j}{\psi_j}$ span $\Herm(\C^2)$. Using the above ideas 
we are able to obtain the lower bound $t_\Delta = 0.8143$, 
which up to two digits of precision matches the value for the critical visibility for the tetrahedron POVM $t_{\mathrm{tetra}}=\sqrt{2/3}$ (see Supplemental Material~\cite{supp}). 

In Supplemental Material we present the results of exploratory studies of application of quasi POVMs for PM-simulability for qutrit POVMs that are covariant with respect to $\mathbb{Z}_3 \times \mathbb{Z}_3$ group \cite{Renes2003}. 
  We expect that quasi-POVMs can be applied in other problems, such as construction of local hidden-variable models for general quantum states~\cite{Dani2015,Hirsch2015} or the problem of joint measurability~\cite{Heinosaari2015,Jessica2016}.

\paragraph*{Improved local-hidden variable model for Werner states---}

In the Bell nonlocality scenario~\cite{Bell1964,Brunner2014}, spatially separated parties perform local measurements on a shared multipartite quantum state.
Some correlations obtained in this form violate the so-called Bell inequalities, which are conditions satisfied by all local hidden-variable (LHV) models. Apart from their fundamental importance, Bell inequalities lie at the center of many non-classical features of quantum cryptography and quantum computation.
Although entanglement is a necessary resource for nonlocality, it is not sufficient, since some entangled states admit LHV models.

Despite the fact that general approaches have been proposed~\cite{Dani2015,Hirsch2015}, the most studied family of states regarding LHV models are the Werner states introduced in~\cite{Werner1989}. While these states are defined for arbitrary dimension, here we focus on the case of two qubits. Werner states then read as $\rho_W(p) = p\Psi_- + (1-p)\I_4/4$, where $\Psi_-$ denotes the projector onto the singlet state and $\I_4$ the identity in dimension four. Initially, Werner proved that these states have a local model under PMs for $p\leq 1/2$. Later, it was shown in~\cite{Acin2006} that the critical noise at which they cease to violate any Bell inequality under projective measurements is equal to the inverse of the Grothendieck constant of order 3, $p^*=1/K_G(3)$. The best known upper bound on this constant can be found in \cite{Geneva2016}, which implies that the critical noise is $p^*\approx 0.68$.
For general measurements, Barrett showed that for $p \leq 5/12 \approx 0.416$ these states have an LHV model~\citep{Barrett2002}.
We now provide an improvement over this result using our formalism.

The result is based on the mentioned self-duality of the depolarizing channel, which was also used in the context of Bell inequalities in~\cite{Bowles2015}. In the bipartite case, it implies that local noisy measurements $\Phi_t(\M _A), \Phi_t(\M _B)$
acting on the two parts of a Werner state $\rho_W(p)$ 
are equal to noise-free measurements $\M_A, \M_B$ acting on $\rho_W(t^2p)$. 
This means that perfect POVMs acting on a Werner state of visibility $2p^*/3$ produce the same statistics as noisy POVMs with visibility $\sqrt{2/3}$ acting on Werner states of visibility $p^*$. But we have shown that at measurement visibility
$t_\Delta$, POVMs can be mimicked by PMs. And at the state visibility $p^*$, Werner states have a local model under these measurements. Putting all these things together, we achieve a new bound $p =t_\Delta^2p^*\approx 0.4519$ for the POVM-locality of two-qubit Werner states, improving over the previous best known result~\citep{Barrett2002}.

\paragraph*{Discussion---}

There is a number of open problems related to the subject of our work. First, it is natural to ask whether for arbitrary dimension the set of PM-simulable measurements can be characterised in terms of SDP programmes and what is the complexity of certification of PM-simulability. Another interesting problem is to identify those POVMs 
that are most robust to PM-simulability, which we conjecture to be symetric-information-complete (SIC) measurements. Another relevant open question is to understand the scaling of the critical value $t(d)$ for projective-simulability. This could be useful to improve the existing local models for Werner states and generalisations of it~\cite{Almeida2010}, as done here for two qubits. Also, the applications of the generalised Naimark theorem still remain to be explored. It is interesting to interpret our results in the context of resource theories. One can consider a resource theory for POVM measurements. In this theory, the resource is the non-projective character of the measurement, projective measurements are the free objects, while mixing and classical processing are the free operations.
 The full construction of the POVM resource theory is an interesting research problem and our work represents the first step in this direction.

\begin{acknowledgments} 
We thank S. Pironio for discussions at early stages of this work. We specially thank T. V\'ertesi for sharing his numerical results with us and discussions.  We also thank N. Brunner, F. Hirsch, M.T. 
 (CoG QITBOX), Axa Chair in Quantum Information Science, 
 Spanish MINECO (FOQUS FIS2013-46768 and Severo Ochoa SEV-2015-0522),
 Fundaci\'{o} Privada Cellex, and Generalitat de Catalunya (SGR 875). P. W acknowledges  computational resources granted by the High Performance Computing Center North (SNIC 2016/1-320). L. G acknowledges CAPES (Brazil). 
 
 {Note added---}During the development of this project we
 became aware of a complementary work by F. Hirsch et al. \cite{Geneva2016} where the authors computed exactly the 
 critical visibility $t(2)=\sqrt{2/3}$ and used it  as here to improve the local model for all POVMs for Werner states.
\end{acknowledgments}

\bibliography{refmeas1}

\begin{thebibliography}{41}%
\makeatletter
\providecommand \@ifxundefined [1]{%
 \@ifx{#1\undefined}
}%
\providecommand \@ifnum [1]{%
 \ifnum #1\expandafter \@firstoftwo
 \else \expandafter \@secondoftwo
 \fi
}%
\providecommand \@ifx [1]{%
 \ifx #1\expandafter \@firstoftwo
 \else \expandafter \@secondoftwo
 \fi
}%
\providecommand \natexlab [1]{#1}%
\providecommand \enquote  [1]{``#1''}%
\providecommand \bibnamefont  [1]{#1}%
\providecommand \bibfnamefont [1]{#1}%
\providecommand \citenamefont [1]{#1}%
\providecommand \href@noop [0]{\@secondoftwo}%
\providecommand \href [0]{\begingroup \@sanitize@url \@href}%
\providecommand \@href[1]{\@@startlink{#1}\@@href}%
\providecommand \@@href[1]{\endgroup#1\@@endlink}%
\providecommand \@sanitize@url [0]{\catcode `\\12\catcode `\$12\catcode
  `\&12\catcode `\#12\catcode `\^12\catcode `\_12\catcode `\%12\relax}%
\providecommand \@@startlink[1]{}%
\providecommand \@@endlink[0]{}%
\providecommand \url  [0]{\begingroup\@sanitize@url \@url }%
\providecommand \@url [1]{\endgroup\@href {#1}{\urlprefix }}%
\providecommand \urlprefix  [0]{URL }%
\providecommand \Eprint [0]{\href }%
\providecommand \doibase [0]{http://dx.doi.org/}%
\providecommand \selectlanguage [0]{\@gobble}%
\providecommand \bibinfo  [0]{\@secondoftwo}%
\providecommand \bibfield  [0]{\@secondoftwo}%
\providecommand \translation [1]{[#1]}%
\providecommand \BibitemOpen [0]{}%
\providecommand \bibitemStop [0]{}%
\providecommand \bibitemNoStop [0]{.\EOS\space}%
\providecommand \EOS [0]{\spacefactor3000\relax}%
\providecommand \BibitemShut  [1]{\csname bibitem#1\endcsname}%
\let\auto@bib@innerbib\@empty
\bibitem [{\citenamefont {Peres}(2006)}]{Peres2006}%
  \BibitemOpen
  \bibfield  {author} {\bibinfo {author} {\bibfnamefont {A.}~\bibnamefont
  {Peres}},\ }\href@noop {} {\emph {\bibinfo {title} {Quantum theory: Concepts
  and methods}}},\ Vol.~\bibinfo {volume} {57}\ (\bibinfo  {publisher}
  {Springer Science \& Business Media},\ \bibinfo {year} {2006})\BibitemShut
  {NoStop}%
\bibitem [{\citenamefont {Renes}\ \emph {et~al.}(2004)\citenamefont {Renes},
  \citenamefont {Blume-Kohout}, \citenamefont {Scott},\ and\ \citenamefont
  {Caves}}]{Renes2003}%
  \BibitemOpen
  \bibfield  {author} {\bibinfo {author} {\bibfnamefont {J.~M.}\ \bibnamefont
  {Renes}}, \bibinfo {author} {\bibfnamefont {R.}~\bibnamefont {Blume-Kohout}},
  \bibinfo {author} {\bibfnamefont {A.~J.}\ \bibnamefont {Scott}}, \ and\
  \bibinfo {author} {\bibfnamefont {C.~M.}\ \bibnamefont {Caves}},\ }\href
  {\doibase 10.1063/1.1737053} {\bibfield  {journal} {\bibinfo  {journal} {J.
  Math. Phys.}\ }\textbf {\bibinfo {volume} {45}},\ \bibinfo {pages} {2171}
  (\bibinfo {year} {2004})}\BibitemShut {NoStop}%
\bibitem [{\citenamefont {Bergou}(2010)}]{Bergou2010}%
  \BibitemOpen
  \bibfield  {author} {\bibinfo {author} {\bibfnamefont {J.~A.}\ \bibnamefont
  {Bergou}},\ }\href {\doibase 10.1080/09500340903477756} {\bibfield  {journal}
  {\bibinfo  {journal} {J. Mod. Opt.}\ }\textbf {\bibinfo {volume} {57}},\
  \bibinfo {pages} {160} (\bibinfo {year} {2010})}\BibitemShut {NoStop}%
\bibitem [{\citenamefont {Derka}\ \emph {et~al.}(1998)\citenamefont {Derka},
  \citenamefont {Buzek},\ and\ \citenamefont {Ekert}}]{Derka1998}%
  \BibitemOpen
  \bibfield  {author} {\bibinfo {author} {\bibfnamefont {R.}~\bibnamefont
  {Derka}}, \bibinfo {author} {\bibfnamefont {V.}~\bibnamefont {Buzek}}, \ and\
  \bibinfo {author} {\bibfnamefont {A.}~\bibnamefont {Ekert}},\ }\href
  {\doibase 10.1103/PhysRevLett.80.1571} {\bibfield  {journal} {\bibinfo
  {journal} {Phys. Rev. Lett.}\ }\textbf {\bibinfo {volume} {80}},\ \bibinfo
  {pages} {1571} (\bibinfo {year} {1998})}\BibitemShut {NoStop}%
\bibitem [{\citenamefont {Bennett}(1992)}]{Bennet2002}%
  \BibitemOpen
  \bibfield  {author} {\bibinfo {author} {\bibfnamefont {C.~H.}\ \bibnamefont
  {Bennett}},\ }\href {\doibase 10.1103/PhysRevLett.68.3121} {\bibfield
  {journal} {\bibinfo  {journal} {Phys. Rev. Lett.}\ }\textbf {\bibinfo
  {volume} {68}},\ \bibinfo {pages} {3121} (\bibinfo {year}
  {1992})}\BibitemShut {NoStop}%
\bibitem [{\citenamefont {Renes}(2004)}]{Renes04Crypt}%
  \BibitemOpen
  \bibfield  {author} {\bibinfo {author} {\bibfnamefont {J.~M.}\ \bibnamefont
  {Renes}},\ }\href {\doibase 10.1103/PhysRevA.70.052314} {\bibfield  {journal}
  {\bibinfo  {journal} {Phys. Rev. A}\ }\textbf {\bibinfo {volume} {70}},\
  \bibinfo {pages} {052314} (\bibinfo {year} {2004})}\BibitemShut {NoStop}%
\bibitem [{\citenamefont {Nielsen}\ and\ \citenamefont
  {Chuang}(2010)}]{Nielsen2010}%
  \BibitemOpen
  \bibfield  {author} {\bibinfo {author} {\bibfnamefont {M.~A.}\ \bibnamefont
  {Nielsen}}\ and\ \bibinfo {author} {\bibfnamefont {I.~L.}\ \bibnamefont
  {Chuang}},\ }\href@noop {} {\emph {\bibinfo {title} {Quantum computation and
  quantum information}}}\ (\bibinfo  {publisher} {Cambridge university press},\
  \bibinfo {year} {2010})\BibitemShut {NoStop}%
\bibitem [{\citenamefont {Jozsa}\ \emph {et~al.}(2003)\citenamefont {Jozsa},
  \citenamefont {Koashi}, \citenamefont {Linden}, \citenamefont {Popescu},
  \citenamefont {Presnell}, \citenamefont {Shepherd},\ and\ \citenamefont
  {Winter}}]{Jozsa2003}%
  \BibitemOpen
  \bibfield  {author} {\bibinfo {author} {\bibfnamefont {R.}~\bibnamefont
  {Jozsa}}, \bibinfo {author} {\bibfnamefont {M.}~\bibnamefont {Koashi}},
  \bibinfo {author} {\bibfnamefont {N.}~\bibnamefont {Linden}}, \bibinfo
  {author} {\bibfnamefont {S.}~\bibnamefont {Popescu}}, \bibinfo {author}
  {\bibfnamefont {S.}~\bibnamefont {Presnell}}, \bibinfo {author}
  {\bibfnamefont {D.}~\bibnamefont {Shepherd}}, \ and\ \bibinfo {author}
  {\bibfnamefont {A.}~\bibnamefont {Winter}},\ }\href@noop {} {\bibfield
  {journal} {\bibinfo  {journal} {Quantum Information \& Computation}\ }\textbf
  {\bibinfo {volume} {3}},\ \bibinfo {pages} {405} (\bibinfo {year}
  {2003})}\BibitemShut {NoStop}%
\bibitem [{\citenamefont {Gisin}(1996)}]{Gisin96}%
  \BibitemOpen
  \bibfield  {author} {\bibinfo {author} {\bibfnamefont {N.}~\bibnamefont
  {Gisin}},\ }\href {\doibase 10.1016/S0375-9601(96)80001-6} {\bibfield
  {journal} {\bibinfo  {journal} {Phys. Lett. A}\ }\textbf {\bibinfo {volume}
  {210}},\ \bibinfo {pages} {151} (\bibinfo {year} {1996})}\BibitemShut
  {NoStop}%
\bibitem [{\citenamefont {V{\'e}rtesi}\ and\ \citenamefont
  {Bene}(2010)}]{Vertesi2010}%
  \BibitemOpen
  \bibfield  {author} {\bibinfo {author} {\bibfnamefont {T.}~\bibnamefont
  {V{\'e}rtesi}}\ and\ \bibinfo {author} {\bibfnamefont {E.}~\bibnamefont
  {Bene}},\ }\href {\doibase 10.1103/PhysRevA.82.062115} {\bibfield  {journal}
  {\bibinfo  {journal} {Phys. Rev. A}\ }\textbf {\bibinfo {volume} {82}},\
  \bibinfo {pages} {062115} (\bibinfo {year} {2010})}\BibitemShut {NoStop}%
\bibitem [{\citenamefont {Ac{\'i}n}\ \emph {et~al.}(2016)\citenamefont
  {Ac{\'i}n}, \citenamefont {Pironio}, \citenamefont {V{\'e}rtesi},\ and\
  \citenamefont {Wittek}}]{Acin2016}%
  \BibitemOpen
  \bibfield  {author} {\bibinfo {author} {\bibfnamefont {A.}~\bibnamefont
  {Ac{\'i}n}}, \bibinfo {author} {\bibfnamefont {S.}~\bibnamefont {Pironio}},
  \bibinfo {author} {\bibfnamefont {T.}~\bibnamefont {V{\'e}rtesi}}, \ and\
  \bibinfo {author} {\bibfnamefont {P.}~\bibnamefont {Wittek}},\ }\href
  {\doibase 10.1103/PhysRevA.93.040102} {\bibfield  {journal} {\bibinfo
  {journal} {Phys. Rev. A}\ }\textbf {\bibinfo {volume} {93}},\ \bibinfo
  {pages} {040102} (\bibinfo {year} {2016})}\BibitemShut {NoStop}%
\bibitem [{\citenamefont {G{\'o}mez}\ \emph {et~al.}(2016)\citenamefont
  {G{\'o}mez}, \citenamefont {G{\'o}mez}, \citenamefont {Gonz{\'a}lez},
  \citenamefont {Ca{\~n}as}, \citenamefont {Barra}, \citenamefont {Delgado},
  \citenamefont {Xavier}, \citenamefont {Cabello}, \citenamefont {Kleinmann},
  \citenamefont {V{\'e}rtesi},\ and\ \citenamefont {Lima}}]{Gomez2016}%
  \BibitemOpen
  \bibfield  {author} {\bibinfo {author} {\bibfnamefont {E.~S.}\ \bibnamefont
  {G{\'o}mez}}, \bibinfo {author} {\bibfnamefont {S.}~\bibnamefont
  {G{\'o}mez}}, \bibinfo {author} {\bibfnamefont {P.}~\bibnamefont
  {Gonz{\'a}lez}}, \bibinfo {author} {\bibfnamefont {G.}~\bibnamefont
  {Ca{\~n}as}}, \bibinfo {author} {\bibfnamefont {J.~F.}\ \bibnamefont
  {Barra}}, \bibinfo {author} {\bibfnamefont {A.}~\bibnamefont {Delgado}},
  \bibinfo {author} {\bibfnamefont {G.~B.}\ \bibnamefont {Xavier}}, \bibinfo
  {author} {\bibfnamefont {A.}~\bibnamefont {Cabello}}, \bibinfo {author}
  {\bibfnamefont {M.}~\bibnamefont {Kleinmann}}, \bibinfo {author}
  {\bibfnamefont {T.}~\bibnamefont {V{\'e}rtesi}}, \ and\ \bibinfo {author}
  {\bibfnamefont {G.}~\bibnamefont {Lima}},\ }\href
  {http://arxiv.org/abs/1604.01417} {\bibfield  {journal} {\bibinfo  {journal}
  {arXiv:1604.01417}\ } (\bibinfo {year} {2016})}\BibitemShut {NoStop}%
\bibitem [{\citenamefont {Vandenberghe}\ and\ \citenamefont
  {Boyd}(1996)}]{sdp}%
  \BibitemOpen
  \bibfield  {author} {\bibinfo {author} {\bibfnamefont {L.}~\bibnamefont
  {Vandenberghe}}\ and\ \bibinfo {author} {\bibfnamefont {S.}~\bibnamefont
  {Boyd}},\ }\href@noop {} {\bibfield  {journal} {\bibinfo  {journal} {SIAM
  Review}\ }\textbf {\bibinfo {volume} {38}},\ \bibinfo {pages} {49} (\bibinfo
  {year} {1996})}\BibitemShut {NoStop}%
\bibitem [{Note1()}]{Note1}%
  \BibitemOpen
  \bibinfo {note} {All algorithms were implemented with PICOS -- Python
  Interface for Conic Optimization Software (available at \protect \url
  {http://picos.zib.de/}) and we made the computational details available at
  \protect \url
  {https://github.com/peterwittek/ipython-notebooks/blob/master/Simulating_POVMs.ipynb}.}\BibitemShut
  {Stop}%
\bibitem [{\citenamefont {Cavalcanti}\ \emph {et~al.}(2016)\citenamefont
  {Cavalcanti}, \citenamefont {Guerini}, \citenamefont {Rabelo},\ and\
  \citenamefont {Skrzypczyk}}]{Dani2015}%
  \BibitemOpen
  \bibfield  {author} {\bibinfo {author} {\bibfnamefont {D.}~\bibnamefont
  {Cavalcanti}}, \bibinfo {author} {\bibfnamefont {L.}~\bibnamefont {Guerini}},
  \bibinfo {author} {\bibfnamefont {R.}~\bibnamefont {Rabelo}}, \ and\ \bibinfo
  {author} {\bibfnamefont {P.}~\bibnamefont {Skrzypczyk}},\ }\href
  {http://arxiv.org/abs/1512.00277} {\bibfield  {journal} {\bibinfo  {journal}
  {Phys. Rev. Lett.}\ }\textbf {\bibinfo {volume} {117}},\ \bibinfo {pages}
  {190401} (\bibinfo {year} {2016})}\BibitemShut {NoStop}%
\bibitem [{\citenamefont {Hirsch}\ \emph
  {et~al.}(2016{\natexlab{a}})\citenamefont {Hirsch}, \citenamefont {Quintino},
  \citenamefont {V{\'e}rtesi}, \citenamefont {Pusey},\ and\ \citenamefont
  {Brunner}}]{Hirsch2015}%
  \BibitemOpen
  \bibfield  {author} {\bibinfo {author} {\bibfnamefont {F.}~\bibnamefont
  {Hirsch}}, \bibinfo {author} {\bibfnamefont {M.~T.}\ \bibnamefont
  {Quintino}}, \bibinfo {author} {\bibfnamefont {T.}~\bibnamefont
  {V{\'e}rtesi}}, \bibinfo {author} {\bibfnamefont {M.~F.}\ \bibnamefont
  {Pusey}}, \ and\ \bibinfo {author} {\bibfnamefont {N.}~\bibnamefont
  {Brunner}},\ }\href {https://arxiv.org/abs/1512.00262} {\bibfield  {journal}
  {\bibinfo  {journal} {Phys. Rev. Lett.}\ }\textbf {\bibinfo {volume} {117}},\
  \bibinfo {pages} {190402} (\bibinfo {year} {2016}{\natexlab{a}})}\BibitemShut
  {NoStop}%
\bibitem [{\citenamefont {Werner}(1989)}]{Werner1989}%
  \BibitemOpen
  \bibfield  {author} {\bibinfo {author} {\bibfnamefont {R.~F.}\ \bibnamefont
  {Werner}},\ }\href {\doibase 10.1103/PhysRevA.40.4277} {\bibfield  {journal}
  {\bibinfo  {journal} {Phys. Rev. A}\ }\textbf {\bibinfo {volume} {40}},\
  \bibinfo {pages} {4277} (\bibinfo {year} {1989})}\BibitemShut {NoStop}%
\bibitem [{\citenamefont {Barrett}(2002)}]{Barrett2002}%
  \BibitemOpen
  \bibfield  {author} {\bibinfo {author} {\bibfnamefont {J.}~\bibnamefont
  {Barrett}},\ }\href {\doibase 10.1103/PhysRevA.65.042302} {\bibfield
  {journal} {\bibinfo  {journal} {Phys. Rev. A}\ }\textbf {\bibinfo {volume}
  {65}},\ \bibinfo {pages} {042302} (\bibinfo {year} {2002})}\BibitemShut
  {NoStop}%
\bibitem [{\citenamefont {Biggerstaff}\ \emph {et~al.}(2009)\citenamefont
  {Biggerstaff}, \citenamefont {Kaltenbaek}, \citenamefont {Hamel},
  \citenamefont {Weihs}, \citenamefont {Rudolph},\ and\ \citenamefont
  {Resch}}]{Exp2009}%
  \BibitemOpen
  \bibfield  {author} {\bibinfo {author} {\bibfnamefont {D.~N.}\ \bibnamefont
  {Biggerstaff}}, \bibinfo {author} {\bibfnamefont {R.}~\bibnamefont
  {Kaltenbaek}}, \bibinfo {author} {\bibfnamefont {D.~R.}\ \bibnamefont
  {Hamel}}, \bibinfo {author} {\bibfnamefont {G.}~\bibnamefont {Weihs}},
  \bibinfo {author} {\bibfnamefont {T.}~\bibnamefont {Rudolph}}, \ and\
  \bibinfo {author} {\bibfnamefont {K.~J.}\ \bibnamefont {Resch}},\ }\href
  {\doibase 10.1103/PhysRevLett.103.240504} {\bibfield  {journal} {\bibinfo
  {journal} {Phys. Rev. Lett.}\ }\textbf {\bibinfo {volume} {103}},\ \bibinfo
  {pages} {240504} (\bibinfo {year} {2009})}\BibitemShut {NoStop}%
\bibitem [{\citenamefont {Medendorp}\ \emph {et~al.}(2011)\citenamefont
  {Medendorp}, \citenamefont {Torres-Ruiz}, \citenamefont {Shalm},
  \citenamefont {Tabia}, \citenamefont {Fuchs},\ and\ \citenamefont
  {Steinberg}}]{FoundChile2011}%
  \BibitemOpen
  \bibfield  {author} {\bibinfo {author} {\bibfnamefont {Z.~E.~D.}\
  \bibnamefont {Medendorp}}, \bibinfo {author} {\bibfnamefont {F.~A.}\
  \bibnamefont {Torres-Ruiz}}, \bibinfo {author} {\bibfnamefont {L.~K.}\
  \bibnamefont {Shalm}}, \bibinfo {author} {\bibfnamefont {G.~N.~M.}\
  \bibnamefont {Tabia}}, \bibinfo {author} {\bibfnamefont {C.~A.}\ \bibnamefont
  {Fuchs}}, \ and\ \bibinfo {author} {\bibfnamefont {A.~M.}\ \bibnamefont
  {Steinberg}},\ }\href {\doibase 10.1103/PhysRevA.83.051801} {\bibfield
  {journal} {\bibinfo  {journal} {Phys. Rev. A}\ }\textbf {\bibinfo {volume}
  {83}},\ \bibinfo {pages} {051801} (\bibinfo {year} {2011})}\BibitemShut
  {NoStop}%
\bibitem [{\citenamefont {Ota}\ \emph {et~al.}(2012)\citenamefont {Ota},
  \citenamefont {Ashhab},\ and\ \citenamefont {Nori}}]{Ota2012}%
  \BibitemOpen
  \bibfield  {author} {\bibinfo {author} {\bibfnamefont {Y.}~\bibnamefont
  {Ota}}, \bibinfo {author} {\bibfnamefont {S.}~\bibnamefont {Ashhab}}, \ and\
  \bibinfo {author} {\bibfnamefont {F.}~\bibnamefont {Nori}},\ }\href {\doibase
  10.1103/PhysRevA.85.043808} {\bibfield  {journal} {\bibinfo  {journal} {Phys.
  Rev. A}\ }\textbf {\bibinfo {volume} {85}},\ \bibinfo {pages} {043808}
  (\bibinfo {year} {2012})}\BibitemShut {NoStop}%
\bibitem [{\citenamefont {Sentis}\ \emph {et~al.}(2013)\citenamefont {Sentis},
  \citenamefont {Gendra}, \citenamefont {Bartlett},\ and\ \citenamefont
  {Doherty}}]{Sentis2013}%
  \BibitemOpen
  \bibfield  {author} {\bibinfo {author} {\bibfnamefont {G.}~\bibnamefont
  {Sentis}}, \bibinfo {author} {\bibfnamefont {B.}~\bibnamefont {Gendra}},
  \bibinfo {author} {\bibfnamefont {S.~D.}\ \bibnamefont {Bartlett}}, \ and\
  \bibinfo {author} {\bibfnamefont {A.~C.}\ \bibnamefont {Doherty}},\ }\href
  {http://stacks.iop.org/1751-8121/46/i=37/a=375302} {\bibfield  {journal}
  {\bibinfo  {journal} {Journal of Physics A: Mathematical and Theoretical}\
  }\textbf {\bibinfo {volume} {46}},\ \bibinfo {pages} {375302} (\bibinfo
  {year} {2013})}\BibitemShut {NoStop}%
\bibitem [{\citenamefont {Buscemi}\ \emph {et~al.}(2005)\citenamefont
  {Buscemi}, \citenamefont {Keyl}, \citenamefont {D'Ariano}, \citenamefont
  {Perinotti},\ and\ \citenamefont {Werner}}]{Buscemi2005}%
  \BibitemOpen
  \bibfield  {author} {\bibinfo {author} {\bibfnamefont {F.}~\bibnamefont
  {Buscemi}}, \bibinfo {author} {\bibfnamefont {M.}~\bibnamefont {Keyl}},
  \bibinfo {author} {\bibfnamefont {G.~M.}\ \bibnamefont {D'Ariano}}, \bibinfo
  {author} {\bibfnamefont {P.}~\bibnamefont {Perinotti}}, \ and\ \bibinfo
  {author} {\bibfnamefont {R.~F.}\ \bibnamefont {Werner}},\ }\href {\doibase
  10.1063/1.2008996} {\bibfield  {journal} {\bibinfo  {journal} {J. Math.
  Phys.}\ }\textbf {\bibinfo {volume} {46}},\ \bibinfo {pages} {082109}
  (\bibinfo {year} {2005})}\BibitemShut {NoStop}%
\bibitem [{\citenamefont {Haapasalo}\ \emph {et~al.}(2012)\citenamefont
  {Haapasalo}, \citenamefont {Heinosaari},\ and\ \citenamefont
  {Pellonp{\"a}{\"a}}}]{Haapasalo2012}%
  \BibitemOpen
  \bibfield  {author} {\bibinfo {author} {\bibfnamefont {E.}~\bibnamefont
  {Haapasalo}}, \bibinfo {author} {\bibfnamefont {T.}~\bibnamefont
  {Heinosaari}}, \ and\ \bibinfo {author} {\bibfnamefont {J.-P.}\ \bibnamefont
  {Pellonp{\"a}{\"a}}},\ }\href {\doibase 10.1007/s11128-011-0330-2} {\bibfield
   {journal} {\bibinfo  {journal} {Quantum Inf. Process.}\ }\textbf {\bibinfo
  {volume} {11}},\ \bibinfo {pages} {1751} (\bibinfo {year}
  {2012})}\BibitemShut {NoStop}%
\bibitem [{Note2()}]{Note2}%
  \BibitemOpen
  \bibinfo {note} {See also~\cite {Buscemi2005} for the systematic treatment of
  {\protect \emph {quantum}} pre- and post-processing of POVMs.}\BibitemShut
  {Stop}%
\bibitem [{\citenamefont {D'Ariano}\ \emph {et~al.}(2005)\citenamefont
  {D'Ariano}, \citenamefont {Presti},\ and\ \citenamefont
  {Perinotti}}]{DAriano2005}%
  \BibitemOpen
  \bibfield  {author} {\bibinfo {author} {\bibfnamefont {G.~M.}\ \bibnamefont
  {D'Ariano}}, \bibinfo {author} {\bibfnamefont {P.~L.}\ \bibnamefont
  {Presti}}, \ and\ \bibinfo {author} {\bibfnamefont {P.}~\bibnamefont
  {Perinotti}},\ }\href {\doibase 10.1088/0305-4470/38/26/010} {\bibfield
  {journal} {\bibinfo  {journal} {J. Phys. A: Math. Gen.}\ }\textbf {\bibinfo
  {volume} {38}},\ \bibinfo {pages} {5979} (\bibinfo {year}
  {2005})}\BibitemShut {NoStop}%
\bibitem [{\citenamefont {Chiribella}\ \emph {et~al.}(2007)\citenamefont
  {Chiribella}, \citenamefont {D'Ariano},\ and\ \citenamefont
  {Schlingemann}}]{Chiribella2007}%
  \BibitemOpen
  \bibfield  {author} {\bibinfo {author} {\bibfnamefont {G.}~\bibnamefont
  {Chiribella}}, \bibinfo {author} {\bibfnamefont {G.~M.}\ \bibnamefont
  {D'Ariano}}, \ and\ \bibinfo {author} {\bibfnamefont {D.}~\bibnamefont
  {Schlingemann}},\ }\href {\doibase 10.1103/PhysRevLett.98.190403} {\bibfield
  {journal} {\bibinfo  {journal} {Phys. Rev. Lett.}\ }\textbf {\bibinfo
  {volume} {98}},\ \bibinfo {pages} {190403} (\bibinfo {year}
  {2007})}\BibitemShut {NoStop}%
\bibitem [{sup()}]{supp}%
  \BibitemOpen
  \href@noop {} {\bibinfo  {journal} {See Supplemental Material at [URL will be
  inserted by publisher]}\ }\BibitemShut {NoStop}%
\bibitem [{\citenamefont {Decker}\ \emph {et~al.}(2005)\citenamefont {Decker},
  \citenamefont {Janzing},\ and\ \citenamefont {R\"otteler}}]{covORTOmes}%
  \BibitemOpen
\bibfield  {journal} {  }\bibfield  {author} {\bibinfo {author} {\bibfnamefont
  {T.}~\bibnamefont {Decker}}, \bibinfo {author} {\bibfnamefont
  {D.}~\bibnamefont {Janzing}}, \ and\ \bibinfo {author} {\bibfnamefont
  {M.}~\bibnamefont {R\"otteler}},\ }\href {\doibase 10.1063/1.1827924}
  {\bibfield  {journal} {\bibinfo  {journal} {J. Math. Phys.}\ }\textbf
  {\bibinfo {volume} {46}},\ \bibinfo {eid} {012104} (\bibinfo {year}
  {2005})}\BibitemShut {NoStop}%
\bibitem [{\citenamefont {Lluis}(2005)}]{Masanes2005}%
  \BibitemOpen
  \bibfield  {author} {\bibinfo {author} {\bibfnamefont {M.}~\bibnamefont
  {Lluis}},\ }\href {https://arxiv.org/abs/quant-ph/0512100} {\bibfield
  {journal} {\bibinfo  {journal} {arXiv:quant-ph/0512100}\ } (\bibinfo {year}
  {2005})}\BibitemShut {NoStop}%
\bibitem [{\citenamefont {Fritz}(2010)}]{Fritz2010}%
  \BibitemOpen
  \bibfield  {author} {\bibinfo {author} {\bibfnamefont {T.}~\bibnamefont
  {Fritz}},\ }\href {http://stacks.iop.org/1367-2630/12/i=8/a=083055}
  {\bibfield  {journal} {\bibinfo  {journal} {New Journal of Physics}\ }\textbf
  {\bibinfo {volume} {12}},\ \bibinfo {pages} {083055} (\bibinfo {year}
  {2010})}\BibitemShut {NoStop}%
\bibitem [{\citenamefont {Chiribella}\ and\ \citenamefont
  {D'Ariano}(2004)}]{Chiribella2004}%
  \BibitemOpen
  \bibfield  {author} {\bibinfo {author} {\bibfnamefont {G.}~\bibnamefont
  {Chiribella}}\ and\ \bibinfo {author} {\bibfnamefont {G.~M.}\ \bibnamefont
  {D'Ariano}},\ }\href {\doibase 10.1063/1.1806262} {\bibfield  {journal}
  {\bibinfo  {journal} {J. Math. Phys.}\ }\textbf {\bibinfo {volume} {45}},\
  \bibinfo {pages} {4435} (\bibinfo {year} {2004})}\BibitemShut {NoStop}%
\bibitem [{\citenamefont {Heinosaari}\ \emph {et~al.}(2015)\citenamefont
  {Heinosaari}, \citenamefont {Kiukas},\ and\ \citenamefont
  {Reitzner}}]{Heinosaari2015}%
  \BibitemOpen
  \bibfield  {author} {\bibinfo {author} {\bibfnamefont {T.}~\bibnamefont
  {Heinosaari}}, \bibinfo {author} {\bibfnamefont {J.}~\bibnamefont {Kiukas}},
  \ and\ \bibinfo {author} {\bibfnamefont {D.}~\bibnamefont {Reitzner}},\
  }\href {\doibase 10.1103/PhysRevA.92.022115} {\bibfield  {journal} {\bibinfo
  {journal} {Phys. Rev. A}\ }\textbf {\bibinfo {volume} {92}},\ \bibinfo
  {pages} {022115} (\bibinfo {year} {2015})}\BibitemShut {NoStop}%
\bibitem [{\citenamefont {{Jessica Bavaresco et al.}}(2016)}]{Jessica2016}%
  \BibitemOpen
  \bibfield  {author} {\bibinfo {author} {\bibnamefont {{Jessica Bavaresco et
  al.}}},\ }\href@noop {} {\bibfield  {journal} {\bibinfo  {journal} {In
  preparation}\ } (\bibinfo {year} {2016})}\BibitemShut {NoStop}%
\bibitem [{\citenamefont {Bell}(1964)}]{Bell1964}%
  \BibitemOpen
  \bibfield  {author} {\bibinfo {author} {\bibfnamefont {J.~S.}\ \bibnamefont
  {Bell}},\ }\href@noop {} {\bibfield  {journal} {\bibinfo  {journal}
  {Physics}\ }\textbf {\bibinfo {volume} {1}},\ \bibinfo {pages} {195}
  (\bibinfo {year} {1964})}\BibitemShut {NoStop}%
\bibitem [{\citenamefont {Brunner}\ \emph {et~al.}(2014)\citenamefont
  {Brunner}, \citenamefont {Cavalcanti}, \citenamefont {Pironio}, \citenamefont
  {Scarani},\ and\ \citenamefont {Wehner}}]{Brunner2014}%
  \BibitemOpen
  \bibfield  {author} {\bibinfo {author} {\bibfnamefont {N.}~\bibnamefont
  {Brunner}}, \bibinfo {author} {\bibfnamefont {D.}~\bibnamefont {Cavalcanti}},
  \bibinfo {author} {\bibfnamefont {S.}~\bibnamefont {Pironio}}, \bibinfo
  {author} {\bibfnamefont {V.}~\bibnamefont {Scarani}}, \ and\ \bibinfo
  {author} {\bibfnamefont {S.}~\bibnamefont {Wehner}},\ }\href {\doibase
  10.1103/RevModPhys.86.419} {\bibfield  {journal} {\bibinfo  {journal} {Rev.
  Mod. Phys.}\ }\textbf {\bibinfo {volume} {86}},\ \bibinfo {pages} {419}
  (\bibinfo {year} {2014})}\BibitemShut {NoStop}%
\bibitem [{\citenamefont {Ac\'{\i}n}\ \emph {et~al.}(2006)\citenamefont
  {Ac\'{\i}n}, \citenamefont {Gisin},\ and\ \citenamefont {Toner}}]{Acin2006}%
  \BibitemOpen
  \bibfield  {author} {\bibinfo {author} {\bibfnamefont {A.}~\bibnamefont
  {Ac\'{\i}n}}, \bibinfo {author} {\bibfnamefont {N.}~\bibnamefont {Gisin}}, \
  and\ \bibinfo {author} {\bibfnamefont {B.}~\bibnamefont {Toner}},\ }\href
  {\doibase 10.1103/PhysRevA.73.062105} {\bibfield  {journal} {\bibinfo
  {journal} {Phys. Rev. A}\ }\textbf {\bibinfo {volume} {73}},\ \bibinfo
  {pages} {062105} (\bibinfo {year} {2006})}\BibitemShut {NoStop}%
\bibitem [{\citenamefont {Hirsch}\ \emph
  {et~al.}(2016{\natexlab{b}})\citenamefont {Hirsch}, \citenamefont {Quintino},
  \citenamefont {V\'ertesi}, \citenamefont {Navascu\'es},\ and\ \citenamefont
  {Brunner}}]{Geneva2016}%
  \BibitemOpen
  \bibfield  {author} {\bibinfo {author} {\bibfnamefont {F.}~\bibnamefont
  {Hirsch}}, \bibinfo {author} {\bibfnamefont {M.~T.}\ \bibnamefont
  {Quintino}}, \bibinfo {author} {\bibfnamefont {T.}~\bibnamefont {V\'ertesi}},
  \bibinfo {author} {\bibfnamefont {M.}~\bibnamefont {Navascu\'es}}, \ and\
  \bibinfo {author} {\bibfnamefont {N.}~\bibnamefont {Brunner}},\ }\href
  {http://arxiv.org/abs/1609.06114} {\bibfield  {journal} {\bibinfo  {journal}
  {arXiv:1609.06114}\ } (\bibinfo {year} {2016}{\natexlab{b}})}\BibitemShut
  {NoStop}%
\bibitem [{\citenamefont {Bowles}\ \emph {et~al.}(2015)\citenamefont {Bowles},
  \citenamefont {Hirsch}, \citenamefont {T\'ulio~Quintino},\ and\ \citenamefont
  {Brunner}}]{Bowles2015}%
  \BibitemOpen
  \bibfield  {author} {\bibinfo {author} {\bibfnamefont {J.}~\bibnamefont
  {Bowles}}, \bibinfo {author} {\bibfnamefont {F.}~\bibnamefont {Hirsch}},
  \bibinfo {author} {\bibfnamefont {M.}~\bibnamefont {T\'ulio~Quintino}}, \
  and\ \bibinfo {author} {\bibfnamefont {N.}~\bibnamefont {Brunner}},\ }\href
  {\doibase 10.1103/PhysRevLett.114.120401} {\bibfield  {journal} {\bibinfo
  {journal} {Phys. Rev. Lett.}\ }\textbf {\bibinfo {volume} {114}},\ \bibinfo
  {pages} {120401} (\bibinfo {year} {2015})}\BibitemShut {NoStop}%
\bibitem [{\citenamefont {Almeida}\ \emph {et~al.}(2010)\citenamefont
  {Almeida}, \citenamefont {Pironio}, \citenamefont {Barrett}, \citenamefont
  {Toth},\ and\ \citenamefont {Ac\'in}}]{Almeida2010}%
  \BibitemOpen
  \bibfield  {author} {\bibinfo {author} {\bibfnamefont {M.~L.}\ \bibnamefont
  {Almeida}}, \bibinfo {author} {\bibfnamefont {S.}~\bibnamefont {Pironio}},
  \bibinfo {author} {\bibfnamefont {J.}~\bibnamefont {Barrett}}, \bibinfo
  {author} {\bibfnamefont {G.}~\bibnamefont {Toth}}, \ and\ \bibinfo {author}
  {\bibfnamefont {A.}~\bibnamefont {Ac\'in}},\ }\href {\doibase
  10.1103/PhysRevLett.99.040403} {\bibfield  {journal} {\bibinfo  {journal}
  {Phys. Rev. Lett.}\ }\textbf {\bibinfo {volume} {99}},\ \bibinfo {pages}
  {040403} (\bibinfo {year} {2010})}\BibitemShut {NoStop}%
\bibitem [{Note3()}]{Note3}%
  \BibitemOpen
  \bibinfo {note} {We have implicitly assumed that $n\geq d$. If this condition
  was not satisfied, the extremal POVMs $\protect \mathbf {N}_k$ would have at
  most $nd\leq d^2$ outcomes. This however does not change our conclusions
  since we can always formally treat $\protect \mathbf {N}_k$ them as $d^2$
  output measurements}\BibitemShut {NoStop}%
\end{thebibliography}%

\onecolumngrid

\part*{Supplemental Material}

\setcounter{equation}{8}

In what follows we present the technical details that complement the main manuscript. In Part~\ref{app:Naimark} we give a detailed proof and a discussion of the generalised Naimark theorem. Subsequently, in Part~\ref{app:tech}, we prove several minor technical results that were stated in the main text. In Part~\ref{app:SDP} , we present explicitly the SDP program deciding the projective-simulability of POVMs for qubits and qutris, as well as the simulability of general POVMs by POVMs with fixed number of outcomes. Finally, in Part~\ref{app:Numerics} we give the details of the numerical techniques used. In, particular we discuss the comprehensive construction of polytopes of ``quasi-POVMs'' for qubits.

\begin{table}[h]
	\begin{centering}
		\begin{tabular}{|c|c|}
			\hline
			\textbf{Symbol} & \textbf{Explanation}\tabularnewline
			\hline
			$\H$ & $d$-dimensional Hilbert space\tabularnewline
			\hline
			$\Herm(\H)$ & Set of Hermitian operators acting on $\H$ \tabularnewline
			\hline
			POVM & Positive operator valued measure \tabularnewline
			\hline
			$\M,\N,\T...$ & Specific POVMs \tabularnewline
			\hline
			$M_i,N_i,T_i...$ & Effects of the corresponding to POVMs \tabularnewline
			\hline
			$\P(d,n)$ & Set of $d$-dimensional POVMs with at most $n$ outcomes \tabularnewline
			\hline
			$\P(d,n;m)$ & Set of $d$-dimensional POVMs with $m\leq n$ outcomes and $n-m$ null effects \tabularnewline
			\hline
			$\PP(d,n)$ & set of $d$-dimensional projective measurements  with at most $n$ outcomes \tabularnewline
			\hline
			$\SP(d,n)$ & Set of projective-simulable POVMs inside $\P(d,n)$ \tabularnewline
			\hline
			$[n]_m$ & $m$-element subsets of the set $\{1,\ldots,n\}$ \tabularnewline
			\hline
			$t(\M)$ & Maximum $t$ for which the POVM $\Phi_t(\M)$ is projective-simulable \tabularnewline
			\hline
			$t(d)$ & Minimum $t$ for which $\Phi_t(\M)$ is projective-simulable for all $\M\in\P(d,n)$ \tabularnewline
			\hline
		\end{tabular}
		
		\par\end{centering}
	
	\caption{Notation used throughout the paper and Supplemental Material}
\end{table}

\appendix

\section{Generalisation of the Naimark theorem for POVMs}\label{app:Naimark}

In this part, we first state and prove the generalised version of Naimark theorem. After the proof, we present an explicit algorithm for the simulation of arbitrary generalised measurements via PM-simulable measurements on a given system together with an ancilla of the same dimension.

\begin{thm}[Generalised Naimark dilation theorem for POVMs]\label{thm:NaimarkTEH} 
	Let $\SP(\H\ot\H', nd)$ be the set of projective-simulable $nd$-outcome POVMs on $\H\ot\H'$, where $\dim(\H)=\dim(\H')=d$. Let $\M\in\P(\H,n)$ be an arbitrary $n$-outcome POVM on $\H$ and let $\kb{\phi}{\phi}$ be a fixed pure state on $\H'$. Then there exists a projective-simulable POVM $\N\in\SP(\H\ot\H', nd)$ such that 
	\begin{equation}\label{eq:simulationTEH}
	\tr(\rho M_i)= \tr(\rho\ot \kb{\phi}{\phi} N_i) 
	\end{equation}
	for $i=1,\ldots,n$, and all states $\rho$ on $\H$. Moreover, $d$ is the minimum possible dimension for ancilla systems with this property.
\end{thm}

\begin{proof}
	Consider the eigendecomposition of the $i$-th effect of $\M$,
	\begin{equation}
	M_i = \sum_{j=1}^d \lambda^{(j)}_{i} \kb{\psi^{(j)}_i}{\psi^{(j)}_i} \ ,\ i=1,\ldots,n \ .
	\end{equation}
	To the POVM $\M$ we can associate the $nd$-output POVM $\M'$ whose effects are constructed from the eigendecomposition of the effects of $\M$, 
	\begin{equation}\label{eq:rank-oneEFF}
	M'_{i,j}= \lambda^{(j)}_{i} \kb{\psi^{(j)}_i}{\psi^{(j)}_i} \ , i=1,\ldots,n\ ,\ j=1,\ldots,d\ .
	\end{equation}
	The original POVM $\M$ can then be realised as a coarse-graining of $\M'$, $\M=\Q\left(\M\right)$, where $\Q$ is a coarse graining strategy specified by the stochastic matrix $q(i|k,j)=\delta_{i,k}$. 
	
	The POVM $\M'$ can be further decomposed into a convex mixture $\M'=\sum p_k \N_k$ of extremal POVMs $\N_k\in\P(d,dn;d^2)$ that have at most $d^2$ outputs~\cite{Chiribella2004} and whose effects are necessary rank-one \footnote{We have implicitly assumed that $n\geq d$. If this condition was not satisfied, the extremal POVMs $\N_k$ would have at most $nd\leq d^2$ outcomes. This however does not change our conclusions since we can always formally treat $\N_k$ them as $d^2$ output measurements}. Thus we have 
	\begin{equation}\label{eq:decompEXTR}
	\M=\Q\left(\sum_k p_k \N_k \right) =\sum_k p_k\Q\left( \N_k \right)\ .
	\end{equation}
	Let us consider a particular POVM $\N_k$ in the above decomposition. In what follows, for the sake of simplicity, we will drop the subscript $k$. Moreover we will assume that effects of $\N$ can be non-zero only for $i=1,\ldots,d^2$ (since we can always relabel the outputs we can assume that without the loss of generality). By linearity, to complete the proof it suffices now to show that for every $\N$ with at most $d^2$ nonzero effects, there exists a projective POVM $\PPP\in\PP(\H\ot\H',d^2)$ such that for every state $\rho$ on $\H$ and every outcome $i=1,\ldots,d^2$ 
	\begin{equation}\label{eq:IMPcond}
	\tr(\rho N_i)= \tr\left(\rho\ot \kb{\phi}{\phi} P_i\right) \ . 
	\end{equation} 
	Let us first observe that rank-one operators $N'_i \otimes \kb{\phi}{\phi}$ form a POVM on a $d$-dimensional Hilbert space $\tilde{\H}=\H\otimes\ket{\psi}$ with \emph{at most} $d^2$ nonzero effects. By the standard Naimark theorem~\cite{Peres2006}, we have that for a POVM with $d^2$ outputs with effects of rank 1 having an orthonormal basis $\{\ket{\psi_i}\}_{i=1}^{d^2}$ such that for all $i$
	\begin{equation}\label{eq:standNAIM}
	N'_i = \PP_{\tilde{\H}} \kb{\psi_i}{\psi_i} \PP_{\tilde{\H}} \ ,
	\end{equation}
	where $\PP_{\tilde{\H}}:\W\rightarrow\W$ is the orthogonal projection projection onto $\tilde{\H}$. Importantly, from the proof of the Naimark theorem given in~\cite{Peres2006}, we know that the construction of orthonormal basis does not depend on the specific structure of $\W$ or $\H'$ but \emph{only} on the number of outputs which in this case equals $d^2$ and a POVM $\N'$. In particular, for $\W=\H\otimes \H'$ we can find the orthonormal basis $\{\ket{\psi_i}\}_{i=1}^{d^2}$ such that for all $i$
	\begin{equation}\label{eq:PARTnaimark}
	N_i \ot \kb{\phi}{\phi}=\left(\I\ot \kb{\phi}{\phi}\right) \kb{\psi_i}{\psi_i} \left(\I\ot \kb{\phi}{\phi}\right)\ .
	\end{equation}
	Setting a projective measurement $P_i=\kb{\psi_i}{\psi_i}$ and using \eqref{eq:PARTnaimark}, we see that $\P$ satisfies equation \eqref{eq:IMPcond} and therefore the proof of which completes the proof of equation \eqref{eq:simulationTEH} is complete.
	
	We now show that $d$ is the minimal dimension of the ancilla $\H'$, necessary for the possibility of simulation of any POVM on $\H$ by some PM-simulable POVM on $\H\ot\H'$. This can be deduced from the existence~\cite{DAriano2005}, for any dimension $d$, of extremal rank-one $d^2$-output POVMs $\M^{\mathrm{ex}}$ on $\H$. Recall that for any dimension $D$ and any number of outputs $n$, the extremal elements of the set of PM-simulable measurements $\SP(D,n)$ are just projective measurements (this follows from Fact 1  in the main text). Consider now the linear mapping
	\begin{equation}
	F:\Herm(\H\ot\H')\ni X\mapsto \tr_{\H}\left[\left(\I\ot \kb{\phi}{\phi}\right) X \left(\I\ot\kb{\phi}{\phi}\right)\right] \in \Herm(\H) \ .
	\end{equation}
	This map can be extended to the linear mapping between spaces of POVMs $\tilde{F}:\P(\H\ot\H',n)\rightarrow \P(\H,n)$. By the extremality of $\M^{\mathrm{ex}}$, linearity of $\tilde{F}$, and the characterisation of extremal elements of $\SP(\H\ot\H',d^2)$, we see that to implement $\M^{\mathrm{ex}}$ via projective PM-simulable measurement on $\H\otimes\H'$, one needs projective measurements with at least $d^2$ elements. Thus $d$ is the minimal dimension of the ancilla Hilbert space. 
\end{proof}

\begin{rem}
	A careful inspection of the proof shows that if a POVM $\M$ has $n\leq d$ outputs, then it can be implemented by PM-simulable POVM on $\H\ot\H'$ where $\dim(\H')=n$. 
\end{rem}

Inspired by this proof, we present the explicit algorithm for finding for a given $\M\in\P(d,n)$ a PM-simulable POVM $\N$ on $\H\ot\H'$.
\begin{enumerate}
	\item Realise $\M\in\P(d,n)$ as a coarse-graining of a rank-one POVM $\M'\in\P(d,nd)$, $\M=\Q(\M')$ (see \eqref{eq:rank-oneEFF});
	\item Decompose $\M'$ onto a convex combination of extremal rank-one POVMs $\N_k$ with \emph{at most} $d^2$ nonzero effects, $\M'=\sum_k p_k \N_k$. This step can be implemented algorithmically (see the general ``method of perturbations'' discussed in~\cite{DAriano2005} and \cite{Sentis2013});
	\item Realise each $\N_k$ by implementing a projective POVM $\PPP_k$ on a system and $d$-dimensional ancilla $\H'$
	\begin{equation}\label{eq:IMPL}
	\tr(\rho \left[\N_k\right]_i)= \tr\left(\rho\ot \kb{\phi}{\phi} \left[\PPP_k\right]_i\right) \ , 
	\end{equation}
	for all outputs $i$ and all states $\rho$ on $\H$. We will focus now on the particular POVM $\N_k$. Also, for the sake of clarity we will skip the subscript $k$. Our construction is a slight modification of the standard proof of the Naimark theorem given for instance in \cite{Peres2006}.  First note that the condition \eqref{eq:IMPL} is equivalent to
	
	\begin{equation}\label{eq:equivalentNAIM}
	\ket{\psi_i} = \ket{v_i}\ket{\phi} + \ket{\psi_{i}^\perp} \ ,\ i=1,\ldots,d^2,
	\end{equation}
	where $\{ \ket{\psi_i}\}_{i=1}^{d^2}$ forms an orthonormal basis of $\H\ot\H'$, $\kb{v_i}{v_i}=N_i$, and vectors $\ket{\psi_{i}^\perp}$ are perpendicular to the subspace $\tilde{H}=\H\otimes\ket{\phi}$.  Consider now a fixed orthonormal basis $\{\ket{w_i}\}_{i=1}^{d^2}$ of $\H\otimes \H'$ having the property that the first $d$ basis vectors spans $\tilde{\H}$ (for instance the standard product basis $\{\ket{e_k}\ket{f_l}\}_{k,l=1}^{d}$, where $\ket{f_1}=\ket{\phi}$. will have this property). Writing equation \eqref{eq:equivalentNAIM} in this notation we get
	\begin{equation}\label{eq:NAIMbasis}
	\ket{\psi_i}=\sum_{j=1}^d a_{ij}\ket{w_j} + \sum_{j=d+1}^{d^2} a_{ij} \ket{w_j} \ ,\ i=1,\ldots,d^2 \ ,
	\end{equation}
	where 
	\begin{equation}\label{eq:NAIMcases}
	a_{ij}=\begin{cases}
	\bra{w_j}\left(\ket{v_i}\ket{\phi}\right)\ & \text{for}\ j=1,\ldots,d \\
	\bk{w_j}{\psi_{i}^\perp} & \text{for}\ j=d+1,\ldots,d^2 
	\end{cases}\ .
	\end{equation}
	
	The condition that operators $N_i$ form a POVM on $\H$ is then equivalent to orthogonality of the first $d$ columns of the matrix $\mathbb{A}_{ij}\defeq a_{ij}$ (in the sense of the standard inner product in $\C^{d^2}$). On the other hand, vectors $\ket{\psi_i}$ form an orthonormal basis of $\H\ot\H'$ if and only if the matrix $\mathbb{A}$ is unitary (which is equivalent to the orthogonality of its columns). Therefore, for every $d^2$-output POVM $\N$ we can construct the desired orthogonal basis $\{\ket{\psi_i}\}_{i=1}^{d^2}$ of $\H\ot\H'$ by complementing the submatrix  $\left(a_{ij}\right)_{i=1,j=1}^{i=d^2,j=d}$ to the unitary $d^2 \times d^2$ matrix $\mathbb{A}$. This can be realised in practice for instance by the application of the Gram-Schmidt process.

\end{enumerate}

\begin{rem}
	The main difference of our POVM simulation scheme and the "standard" methods of simulation of POVM  via projective measurements on a system extended by an ancilla \cite{Peres2006,Jozsa2003,covORTOmes}. In the standard approaches it is assumed that states $\ket{\psi_{i}^{\perp}}$ from \eqref{eq:equivalentNAIM}  have the structure $\ket{\psi_{i}^{\perp}}=\ket{s_{i}}\ket{f_i}$, where $\ket{f_i}\perp\ket{\phi}$. A careful inspection of the proof of Naimark theorem \cite{Peres2006} showed that this is structure is not necessary. If we had kept this structure we would be forced to assume (by the extermality arguments) that the dimension of the ancilla scales like $d^2$, which is more demanding then our construction.
\end{rem}

\section{Proofs of technical results}\label{app:tech}

In this part, we state and prove a number of technical results that were used in the main text or stated without a complete proof.

\subsection*{Partial characterisations of PM-simulability}

\begin{lem}\label{lem:suffTECH}
	Let $\M\in\P(d,n)$ be an $n$-output POVM on the Hilbert space $\mathcal{H}$ of dimension $d$. Let 
	$\lambda_i^{max}$ denote the maximal eigenvalue of the $i$-th effect $M_i$. If for some $k\in\left\{ 1,\ldots,n\right\}$ it holds that
	\begin{equation}
	\sum_{i\neq k}\lambda_i^{\max}\leq1\ ,\label{eq:conditionTECH}
	\end{equation}
	then
	$\M$ is PM-simulable. 
\end{lem}

\begin{proof}
	Without the loss of generality, we can assume that $k=1$. By the assumption \eqref{eq:conditionTECH}, we have
	\begin{equation}
	\sum_{i=2}^{N} \lambda^{\max}_{i} = 1-\delta, \end{equation}
	for some $\delta\geq 0$.
	Hence,
	\begin{align}
	M_{1} & = \I - \sum_{i=2}^{n}{M_i}\\
	&=\left(\sum_{i=2}^{n}\lambda^{\max}_{i}+\delta\right) \mathbb{I}-\sum_{i=2}^{N}M_{i} \\
	&= \delta\left(\I,0,\ldots,0\right)+\sum_{i=2}^{n}(\lambda^{\max}_{i}\I - M_i).
	\end{align}
	The definition of $\lambda^{\max}_{i}$ ensures $\lambda^{\max}_{i}\I - M_i \geq 0$, and therefore the dichotomic POVMs $\mathcal{Q}_i$, whose effects are
	\begin{equation}
	\left[\N_{i}\right]_{k}= \begin{cases}
	\I-M_k/\lambda^{\max}_{i} & \text{if}\,k=1\\
	M_k/\lambda^{\max}_{i} & \text{if}\,k=i\\
	0 & \text{otherwise}
	\end{cases}
	\end{equation}
	are well-defined. Then, we have the convex decomposition $\M=\delta\left(\mathbb{I},0,\ldots,0\right)+\sum_{i=2}^{n}\lambda^{\max}_{i} \N_{i}$. Since two-outcome POVMs are projective-simulable \cite{Masanes2005,Fritz2010}, this concludes the proof. 
\end{proof}

\begin{lem}[Lower bound for the critical visibility in arbitrary dimension]\label{lem:lowerboundTECH}
	For every $\M\in\P\left(d,n\right)$, its depolarised version $\Phi_{\frac{1}{d}} \left(\M\right)$ is PM simulable.
	Consequently, for any dimension $d$ we have $t(d)\geq 1/d$. 
\end{lem}
\begin{proof}[Final step of the proof of Lemma 1 from the main text] As remarked in the main text, we can restrict our attention to POVMs with rank-one effects. Let $\M\in\P(d,n)$ be such that $M_i = \alpha_i \Pi_i$, where $\sum_{i=1}^n \alpha_i = d$ and $\Pi_i$ are rank-one projectors.
	A natural candidate for a protocol to projective-simulate $\M$ consists of the following steps:
	\begin{enumerate}
		\item Choose $i\in\{1,...,n\}$ with probability $\alpha_i/d$;
		\item Perform the projective measurement $(\Pi_i, \I - \Pi_i)$;
		\item If the outcome corresponds to $\Pi_i$, output $i$;
		\item If the outcome corresponds to $\I-\Pi_i$, output any $j\in\{1,...,n\}$ with probability $\alpha_j/d$.
	\end{enumerate}
	Implementing the above protocol corresponds to performing a measurement $\N$ with effects given by
	\begin{eqnarray}
	N_i &=& \frac{\alpha_i}{d}\left[\Pi_i + (\I - \Pi_i)\frac{\alpha_i}{d} \right] + \left[\sum_{j \neq i}{\frac{\alpha_j}{d}(\I - \Pi_j)}\right]\frac{\alpha_i}{d} \nonumber \\ 
	&=& \frac{1}{d}\alpha_i\Pi_i + (1-\frac{1}{d})\alpha_i\frac{\I}{d}\\
	&=& \Phi_{\frac{1}{d}}(M_i). \nonumber
	\end{eqnarray}
	
\end{proof}

\begin{rem}
	For the particular case where $\M$ has rank-one effects with trace $1/d$ (such as $\M^{\mathrm{tetra}}$) and other symmetric informationally complete POVMs), outputting $j\neq i$ with uniform probability $1/(d^2-1)$ in step 4 of the protocol above yields the simulation of $\Phi_{\frac{d}{d^2-1}}\left(\M\right)$. 
\end{rem}

\subsection*{Optimal projective simulation of $\M^{\mathrm{tetra}}(\sqrt{2/3})$}

The symmetries of tetrahedral POVMs $\M^{\mathrm{tetra}}\in\P(2,4)$ allow us to find a projective simulation for its depolarised version $\M^{\mathrm{tetra}}(\sqrt{2/3})$, shown to be optimal by the SDP based on equation \eqref{eq:conv1}.

In Example~\ref{convhull vertesi} we saw that only six projective measurements $\Pi_{ij},\ ij = 12,...,34,$ suffice to carry out the simulation.
The uniform distribution of the Bloch vectors of the effects suggests that the PMs should have the direction of the bisectrices between each pair of effects, given by
\begin{eqnarray}
\Pi_{ij}^i &=& \frac12\left(\I + \sqrt{\frac{3}{2}}(\vec{n}_i - \vec{n}_j)\cdot \vec{\sigma}\right)\\
\Pi_{ij}^j &=& \frac12\left(\I + \sqrt{\frac{3}{2}}(\vec{n}_j - \vec{n}_i)\cdot \vec{\sigma}\right),
\end{eqnarray}
where $\sqrt{3/2}$ is a normalisation constant corresponding to the inverse of the length of the edge of a tetrahedron inscribed in the unit sphere. Explicit computation shows that we see that if we choose one of the projective measurements $\Pi_{ij}$ with uniform probability $1/6$ and output the corresponding outcome $i$ or $j$, we will be simulating the POVM $\Phi_{\sqrt{\frac{2}{3}}}\left( \M^{\mathrm{tetra}}\right)$.

\subsection*{Characterisations of PM-simulable and m-chotomic-simulable POVMs}

\begin{lem}[Characterisation of the convex hull of $m$-chotomic measurements for arbitrary dimension]\label{lem:mOUTchar}
	Let $\M\in\P(d,n)$ be $n$-outcome measurement on $d$ dimensional Hilbert space. Let $m\leq n$. $\M\in\P(d,n;m)^{\mathrm{conv}}$ if and only if 
	\begin{equation}\label{eq:m-simulability}
	\M=\sum_{X\in[n]_m}p_{X}\N_{X}\ 
	\end{equation}
	where 
	\begin{itemize}
		\item the sum in \eqref{eq:m-simulability} is over m-element subsets of $\left\{ 1,\ldots,n\right\}$;
		\item $\N_X\in\P(d,n;m)$ have non-zero effects only for outputs belonging to $X$, i.e., 
		\begin{equation}
		\left[\N_X\right]_i=0\ \ \text{for}\ i\notin X;
		\end{equation}
		\item $\{p_{X}\}_X$ form a probability distribution on $m$-element subsets of $\left\{ 1,\ldots,n\right\}$, 
		\begin{equation}
		p_{X}\geq0\,,\,\sum_{X\in[n]_m}p_{X}=1\ \ .
		\end{equation}
		
	\end{itemize} 
	
\end{lem}

\begin{proof}
	Let us first note that every measurement $\M\in\P(d,n)$ that can be expressed by \eqref{eq:m-simulability} belongs to the convex hull of $m$-outcome measurements, $\M\in\P(d,n;m)^{\mathrm{conv}}$. In order to show that every $\M\in \P(d,n;m)^{\mathrm{conv}}$ can be decomposed according to equation \eqref{eq:m-simulability} we use a convex  decomposition  
	\begin{equation}
	\M=\sum_k p_k \M_k , \
	\end{equation}
	where $\{p_k\}$ is a probability distribution and $\M_k \in \M\in \P(d,n;m)$.  Grouping the POVMs appearing in the above sum according to the subsets $X\in[n]_m$, where their effects can be nonzero, we get 
	\begin{equation}
	\M=\sum_{X\in[n]_m} \sum_{k:\left[\M_k\right]_i\neq 0\ \text{iff}\ i\in X} p_k \M_k \ .   
	\end{equation}
	Defining
	\begin{gather} 
	\tilde{p}_X\defeq \sum_{k:\left[\M_k\right]_i\neq 0\ \text{iff}\ i\in X} p_k \ , \\
	\tilde{p}_X \tilde{\M}_X \defeq  \sum_{k:\left[\M_k\right]_i\neq 0\ \text{iff}\ i\in X} p_k \M_k\ ,
	\end{gather}
	we finally get the desired decomposition 
	\begin{equation}
	\M=\sum_{X\in[n]_m}= \tilde{p}_X \tilde{\M}_X \ .
	\end{equation}

\end{proof}

We now move to the characterisation of PM-simulable measurements for qutrits. Let us first sate a useful mathematical result that will be used in the the proof of Lemma \ref{lem:mainresult2TEH} below.

\begin{lem}\label{lem:EXTqutrit}
	Let $\P_1 (3,3)$  be the set of qutrit 3-outcome measurements whose effects satisfy $\tr(M_i)=1$. The set of extremal points of this convex set consist solely of qutrit rank 1 projective measurements $\PPP\in\PP(3,3)$ i.e. measurements of the form $\PPP=\left(\kb{\psi_1}{\psi_1},\kb{\psi_2}{\psi_2},\kb{\psi_3}{\psi_3}\right)$, where $\{\ket{\psi_i}\}_{i=1}^3$ is an orthonormal basis of $\C^3$.
\end{lem}

\begin{proof}
	
	We apply the analogue of  ``method
	of perturbations'' \cite{DAriano2005} that allows to check whether a given POVM $\mathcal{M}\in\mathcal{P}\left(\mathcal{H},n\right)$
	is extremal or not. For completeness let us introduce the basics of the method of perturbations. 
	
	A POVM $\mathcal{M}\in\mathcal{P}\left(\mathcal{H},n\right)$
	\emph{is not} extremal if and only if there exits a vector of hermitian
	operators $\X=\left(X_{1},\ldots,X_{n}\right)\in\mathrm{Herm}\left(\mathcal{H}\right)^{\times n}$
	such that for every output $i$
	\begin{equation}
	\mathrm{supp}\left(X_{i}\right)\subset\mathrm{supp}\left(M_{i}\right)\,\,\label{eq:supports}
	\end{equation}
	where $\mathrm{supp}(X)$ denotes the support of the hermitian operator $X$,  and
	\begin{equation}\label{eq:linear dependance}
	\sum_{i=1}^{n}X_{i}=0\ ,
	\end{equation}
	Existence of the perturbation $\X$ allows to construct POVMs 
	\begin{equation}\label{eq:perturbation}
	\M_{I}=\M+\epsilon\X\,,\,\M_{II}=\M-\epsilon\X \ ,
	\end{equation}
	for sufficiently small $\epsilon>0$. From (\ref{eq:perturbation})
	it follows that $\M=\frac{1}{2}\M_{I}+\frac{1}{2}\M_{II}$.

	We adopt a method of perturbation to the set $\P_1(3,3)$. By the analogous construction as before we can show that $\M\in P_1(3,3)$ \emph{is not} extremal element in $P_1 (3,3)$ if and only if there exist a perturbation $\X=(X_1,X_2,X_3)$ such that for all outputs $i$: (I) $\mathrm{supp}\left(X_{i}\right)\subset\mathrm{supp}\left(M_{i}\right)$, (II) $\tr(X_i)=0$, and (III) $\sum_{i=1}^3 X_i=0$ (the condition (II) follows from the fact that the traces of effects have to be preserved upon the addition of the perturbation). W now introduce \emph{a class} of a POVM $\M\in P_1 (3,3)$ as a vector of integers $r(\M)=(r_1,r_2,r_3)$, where $r_i=\mathrm{rank}(M_i)$ (permutations in the order of the ranks correspond to relabelling). In what follows we show that whenever $r(\M)\neq(1,1,1)$ there always exist a perturbation $\X$ satisfying the above conditions (I-III):
	
	\begin{itemize}
		\item Cases $r(\M)=(a,b,0)$, where $a,b$- arbitrary, are not possible as $\tr(M_3)=1\neq0$;
		\item  Cases $r(\M)=(a,b,1)$, where $a,b> 2$ are impossible because if $M_3$ has rank one and $\tr(M_3)=1$ we necessary have $M_3=\kb{\psi}{\psi}$, for some pure state $\kb{\psi}{\psi}$. Therefore, ranks of the reaming effect have to be at most $2$;  
		\item The case $r(\M)=(2,1,1)$ is impossible for analogous reasons;
		\item If $r(\M)=(3,a,b)$, where $a,b\leq$ we can take  perturbation $\X$ of the form 
		\begin{equation}\label{eq:pertPROOF}
		\X = \left(\kb{\psi_1}{\psi_2}+\kb{\psi_2}{\psi_1}, - \kb{\psi_2}{\psi_1}-\kb{\psi_1}{\psi_2}, 0\right)\ ,
		\end{equation}
		where $\ket{\psi_1},\ket{\psi_2}$ are eigenvalues of $M_2$ corresponding to nonzero eigevectors;
		\item If $r(\M)=(2,2,2)$ we consider spaces $\mathrm{S}_i = \Herm_0 \left[\mathrm{supp}\left(M_i\right)\right]$, 
		where $\Herm_0(\mathcal{V})$ denotes the real vector space spece of traceless hermitian matrices supported on the Hilbert space $\mathcal{V}$. Since  $\mathrm{supp}(M_i)=2$ we have $\dim(\mathrm{S}_i) =3$ (these spaces are isomorphic to the real space spanned by the Pauli matrices in the qubit case). On the other hand $\Herm_0(\C^3)=8$ and and $\mathrm{S}_i \subset  \Herm_0(\C^3)$. Therefore there must be some linear dependency between elements of $\mathrm{S}_i$. In other words there exist operators $S_i \in \mathrm{S}_i$ such that 
		\begin{equation}
		S_1+S_2+S_3=0 \ .
		\end{equation}
		Therefore, $\X=(S_1,S_2,S_3)$ is a valid perturbation for $\M$.
		\item The case $r(\M)=(2,2,1)$ is easy to analyse as by  $\tr(M_3)=1$ we have $M_3=\kb{\psi}{\psi}$, for some pure state $\kb{\psi}{\psi}$. Then necessary operators $M_1, M_2$ commute and have the same supports. Consequently, perturbation given in \eqref{eq:pertPROOF} works.
		
	\end{itemize}

\end{proof}

\begin{lem}[Projective-simulable measurements for qutrits]
	\label{lem:mainresult2TEH}
	For $d=3$, projective measurements simulate 2-outcome measurements and 3-outcome measurements with trace-1 effects. Moreover we have
	\begin{equation}\label{eq:crucialfact2}
	\SP\left(3,n\right) = \left(\P\left(3,n;2\right) \cup \P_1\left(3,n;3\right)\right)^{\mathrm{conv}}\ \ . 
	\end{equation}
	where $\P_1\left(3,n;3\right) \subset \P\left(3,n\right)$ denotes the set of three output measurements on $\C^3$  with effects having unit trace.
\end{lem}

\begin{proof} 
	
	Let us first note that we can resort to POVMs having at most  three outputs. Indeed if we first show that
	\begin{equation}\label{eq:crucialfact2THREE}
	\SP\left(3,3\right) = \left(\P\left(3,3;2\right) \cup \P_1\left(3,3\right)\right)^{\mathrm{conv}}\ \ ,
	\end{equation}
	then equation \eqref{eq:crucialfact2} will follow from it because:
	\begin{itemize}
		\item Every PM-simulable POVM $\M\in\SP(3,n)$ can be obtained as a convex combination of POVMs $\M_X$, where $\M_X$ is a PM-simulable measurement with \emph{at most} 3 nonzero effects corresponding to the outputs laying in $X\in[n]_3$;
		\item Similarly, every POVM $\M\in\left(\P\left(3,n;2\right) \cup \P_1\left(3,n;3\right)\right)^{\mathrm{conv}}$ can be written as a convex combination of POVMs $\N_X$, for $X\in[n]_3$. Each $\N_X$ can be expressed as a convex combination of elements form $\P\left(3,3;2\right) \cup \P_1\left(3,3\right)$ and has \emph{at most} 3 nonzero effects corresponding to outputs $i\in X$.
	\end{itemize}
	Therefore, from now on we will focus on showing \eqref{eq:crucialfact2THREE}. Since $\SP(3,3)=\PP(3,3)^{\mathrm{conv}}$ we see that  $\SP(3,3)\subset \left(\P\left(3,3;2\right) \cup \P_1\left(3,3\right) \right)^{\mathrm{conv}}$ because
	\begin{itemize}
		\item Two outcome projective measurements are as special case of two outcome measurements;
		\item Rank-one qutrit projective measurements belong to $\P_1 (3,3)$;
	\end{itemize}
	
	As two-outcome measurements $\P(3,3;2)$ are simulable by projective measurements, in order to show the reverse inclusion $\SP(3,3)\supset \left(\P\left(3,3;2\right) \cup \P_1\left(3,3\right) \right)^{\mathrm{conv}}$, it suffices to know that every $\M\in\P_1 (3,3)$ can be expressed as a convex combination of projective rank-one qutrit  measurements . However, this is precisely the content of Lemma \ref{lem:EXTqutrit} which we have proved before.

\end{proof}

\begin{lem}[Efficient characterisation of PM-simulable measurements for qutrits] \label{lem:mainresult2aux}
	Let $\M\in\P(3,n)$ be an $n$-outcome measurement on a qutrit Hilbert space. $\M\in\left(\P\left(3,n;2\right) \cup \P_1\left(3,n;3\right)\right)^{\mathrm{conv}}$ if and only if 
	
	\begin{equation}\label{eq:qutritChAR}
	\M=\sum_{Y\in[n]_2}p_{X}\N_{X} +\sum_{X\in[n]_3}p_{Y}\N_{Y}\ ,
	\end{equation}
	where 
	
	\begin{itemize}
		\item the sums in \eqref{eq:qutritChAR} are over 2- and 3-element subsets of $\left\{ 1,\ldots,n\right\}$;
		\item $\N_X\in\P(3,n;2)$ can have non-zero effects only for outputs belonging to $X$, i.e., 
		\begin{equation}
		\left[\N_X\right]_i=0\ \ \text{for}\ i\notin X;
		\end{equation}
		
		\item $\N_Y\in\P_1(3,n;3)$ have effects of unit trace and can have non-zero effects only for outputs belonging to $Y$, i.e., 
		\begin{equation}
		\left[\N_Y\right]_i=0\ \ \text{for}\ i\notin X;
		\end{equation}

		\item $\{p_{X}\}\cup \{p_{Y}\}$ form a probability distribution, i.e., $p_{X}\geq 0$, $p_{Y}\geq 0$ and 
		\begin{equation}
		\sum_{X\in[n]_2}p_{X} +\sum_{Y\in[n]_3}p_{Y}=1\ .
		\end{equation} 
		
	\end{itemize}
	
\end{lem}

\begin{proof}
	The proof follows form Lemma \ref{lem:mainresult2TEH} and is completely analogous to the proof of Lemma \ref{lem:mOUTchar}. The only difference is that we have to consider two- and three-element subsets that support different types of POVMs.
\end{proof}


A slight modification of the above results yields SDPs for the computation of the critical visibility $t(\M)$ in the cases where $\M\in\P(d,n)$, for $d=2,3$.

\subsection*{Counterexamples for the generalisations of PM-simulation constructions for higher dimensions}

\begin{fact}\label{fact:doutp}
	For dimension $d>2$ there exist $d$-outcome POVMs $\mathcal{M}\in\mathcal{P}\left(d,d\right)$
	that are not PM-simulable. \end{fact}

\begin{proof}
	First consider the case $d=3$ and take the modified trine measurement~\cite{Jozsa2003} on $\C^3$ 
	\begin{equation}
	\M\defeq\left(\frac{2}{3}\kb{\psi_1}{\psi_1}, \frac{2}{3}\kb{\psi_2}{\psi_2},\frac{2}{3}\kb{\psi_3}{\psi_3} +\kb{2}{2}\right) \ ,
	\end{equation}
	where $\ket{\psi_j}=\cos(\pi j /3 ) \ket{0}+ \sin(\pi j /3 ) \ket{1}$. $\M$ turns out to by extremal in $\P(d,d)$ and thus cannot be PM-simulable. The extremality of $\M$ follows from the the extremality of the trine POVM on $\C^2$
	\begin{equation}
	\M_{\mathrm{trine}} = \left(\frac{2}{3}\kb{\psi_1}{\psi_1}, \frac{2}{3}\kb{\psi_2}{\psi_2},\frac{2}{3}\kb{\psi_3}{\psi_3}\right)\ .
	\end{equation}
	Indeed if $\M$ was not extremal then $\M_{\mathrm{trine}}$ cannot be extremal since 
	\begin{equation}
	\M= \PP\left(\M_{\mathrm{trine}}\right) \ ,
	\end{equation}
	where $\PP$ is the orthogonal projector onto two-dimensional subspace spanned by $\ket{0}$ and $\ket{1}$.
	
	Extremal non-projective POVMs measurements in $\P(d,d)$ can be constructed from $\M$ by supplementing it
	with the complete projective measurement in the orthogonal complement of $\mathrm{span}_\C \lbrace\ket{0},\ket{1},\ket{2}\rbrace$
	distributed among the remaining $d-3$ outputs. 
	
\end{proof}

\begin{lem}\label{fact:tracone}
	For dimension $d>3$, there exist $d$ output POVMs $\M\in\P(d,d)$ with effects having unit trace that are not PM-simulable.
\end{lem}

\begin{proof}
	First consider the case $d=4$ and take two qubit tetrahedron measurements
	\begin{equation}
	\T=\left(T_{1},T_{2},T_{3},T_{4}\right)\,,\,\T'=\left(T'_{1},T'_{2},T'_{3},T'_{4}\right)\,,\label{eq:two tetrachedra}
	\end{equation}
	each supported on two mutually orthogonal two-dimensional subspaces of $\mathbb{C}^{4}$
	(denoted by $\W$ and $\W'$ respectively).
	Using $\T$ and $\T'$ we can define a measurement
	on $\mathbb{C}^{4}$,
	\begin{equation}
	\M\defeq\left(T_{1}+T'_{1},T_{2}+T'_{2},T_{3}+T'_{3},T_{4}+T'_{4}\right)\,.\label{eq:construction}
	\end{equation}
	Clearly, $\M$ is not a projective measurement and $\tr(M_i)=1$. Moreover, it is easy to see that $\M$ is extremal. This follows form the
	extremality of measurements $\T$ and $\T'$.
	If $\hat{\T}$ was not extremal, then POVMs $\T,\T'$
	would not be extremal as well since 
	\begin{equation}\label{eq:recovery}
	\T=\PP_{\W}\left(\M\right) ,\ \T'=\PP_{\W'}\left(\M\right)\ ,
	\end{equation}
	where $\PP_{\W}$ and $\PP_{\W'}$
	denote orthogonal projectors onto $\W$ and $\W'$
	respectively. 
	
	Extremal non-projective POVMs in $\mathcal{P}\left(\mathcal{H},d\right)$ having effects of unit trace
	for $d>4$ can be constructed from $\M$ by supplementing it
	with the projective measurement in the orthogonal complement of $\W \oplus \W'$
	distributed among the remaining $d-4$ outputs. 
\end{proof}

\section{SDP programs}\label{app:SDP}

In this part we provide explicit semi-definite programs for deciding PM-simulability for qubits and qutrits, as well as for the simulation of POVMs on the Hilbert space of dimension $d$ by $m$-chotomic measurements. Moreover, we give the modified versions of these programs that allow to compute, for a given POVM $\M$, the critical visibility $t(\M)$ for PM- or $m$-chotomic simulability.

Using Lemma~\ref{lem:mOUTchar}, it is possible to characterise the convex hull of $m-$outcome measurements. 

\begin{lem}[SDP for the convex hull of $m$-outcome measurements]\label{lem:mOUTcharSDP}
	Let $\M\in\P(d,n)$ be an $n$-outcome measurement on a $d$-dimensional Hilbert space. Let $m\leq n$. The question whether $\M\in\P(d,n;m)^{\mathrm{conv}}$ can be solved by the following SDP feasibility problem
	\begin{align}
	& \text{Given} && \M = (M_1,\ldots,M_n)\nonumber \\
	& \text{Check if there exist} &&\left\{\N_X\right\}_{X\in[n]_m}\ ,\ \left\{p_X\right\}_{X\in[n]_m}\ , \nonumber \\
	& \text{Such that} && \M=\sum_{X\in[n]_m} \N_X\ , \nonumber \\
	& && [\N_X]_i \geq 0\ , i=1,\ldots,n \ , \nonumber \\
	& && [\N_X]_i = 0\ \ \text{for}\ i\notin{X} \ , \\
	& && \sum_{i=1}^n [\N_X]_i= p_X \I\ , \nonumber \\
	& && p_X \geq 0\ , \sum_{X\in[n]_m} p_X=1\ . \nonumber
	\end{align}
\end{lem}

\begin{proof}
	This result follows directly from Lemma \ref{lem:mOUTchar}. Having a decomposition \eqref{eq:m-simulability} we can define auxiliary "probability variables" $p_X$ and "rescaled"  effects $\tilde{\N}_X$ defined by the condition $\tilde{\N}_X \defeq p_X \N_X$. Under this notation equation \eqref{eq:m-simulability} becomes exactly equivalent to \eqref{eq:m-simulability} for $\N_X \equiv \tilde{\N}_X$.
\end{proof}

Below we present an explicit example of the SDP description of $\P\left(2,n;2\right)^{\mathrm{conv}}$ for $n=4$ outputs. This case is especially interesting since for qubits the extremal POVMs have at most $4$ outcomes.

\begin{exa}
	\label{convhull vertesi} A four outcome qubit POVM $\M \in\P(2,4)$ is PM-simulable
	if and only if 
	\begin{gather}\
	M_{1}=N_{12}^{+}+N_{13}^{+}+N_{14}^{+}\,,\nonumber \\
	M_{2}=N_{12}^{-}+N_{23}^{+}+N_{24}^{+}\,,\nonumber \\
	M_{3}=N_{13}^{-}+N_{23}^{-}+N_{34}^{+}\,,\label{eq:conv1}\\
	M_{4}=N_{14}^{-}+N_{24}^{-}+N_{34}^{-}\,,\nonumber 
	\end{gather}
	where Hermitian operators $N_{ij}^{\pm}$ satisfy $N_{ij}^{\pm}\geq0$
	and $N_{ij}^{+}+N_{ij}^{-}=p_{ij}\mathbb{I}$ for $i<j$ , $i,j=1,2,3,4$, where $p_{ij}\geq0$ and $\sum_{i<j}p_{ij}=1$.
\end{exa}

The conditions satisfied by variables $N_{ij}^{\pm}$, $p_{ij}$ are of the form of semi-definite
constrains. Notice that the equation \eqref{eq:conv1} actually presents a decomposition of a POVM $\M$ onto six dichotomic measurements. Since two outcome measurement can be simulated by a single projective measurement, we see that arbitrary four-outcome qubit PM-simulable $\M$ can be simulated by at most six projective measurements. Dichotomic measurements with $N'^{\pm}_{ij}=N^\pm_{ij}/p_{ij}$ can be extracted directly form the solution of the SDP programme. The SDP \eqref{eq:conv1} can be modified to calculate the critical visibility $t(\M)$ for a fixed 4-outcome qubit POVM. The conditions satisfied by variables $N_{ij}^{\pm}$, $p_{ij}$ are of the form of semi-definite
constrains. Notice that the equation \eqref{eq:conv1} actually presents a decomposition of a POVM $\M$ onto six dichotomic measurements. 

In analogy to the case of PM-simulability we can define the notion of simulability by $m$-outcome measurements. We analogously define the notion of critical visibility for simulation in terms of $m$-chotomic measurements
\begin{equation}\label{eq:critVisM}
t^{(m)}(\M)\defeq \mathrm{max}\SET{t}{\Phi_t (\M) \in \P(d,n;m)^{\mathrm{conv}}}\ .
\end{equation}
A simple modification of the above result allows us to construct SDP for computing, for the given POVM $\M$, the critical visibility $t^{(m)}(\M)$.

\begin{lem}[SDP for critical visibility for $m$-outcome simulability]\label{lem:mOUTvisSDP}
	Let $\M\in\P(d,n)$ be $n$-outcome measurement on $d$ dimensional Hilbert space. Let $m\leq n$. The critical visibility $t^{(m)}(\M)$ can be computed via the following SPD programme
	\begin{align}
	& \text{Given} && \M = \left(M_1,\ldots,M_n\right)\nonumber \\
	& \text{Maximise} && t\ \nonumber \\
	& \text{Such that} && t\M+(1-t)\left(\tr(M_1)\frac{\I}{d},\ldots,\tr(M_n)\frac{\I}{d}\right)=\sum_{X\in[n]_m} \N_X\ , \nonumber \\
	& && \left\{\N_X\right\}_{X\in[n]_m}\ ,\ \left\{p_X\right\}_{X\in[n]_m}\ , \nonumber \\ 
	& && [\N_X]_i \geq 0\ , i=1,\ldots,n \ , \nonumber \\
	& && [\N_X]_i = 0\ \ \text{for}\ i\notin{X} \ , \\
	& && \sum_{i=1}^n [\N_X]_i= p_X \I\ , \nonumber \\
	& && p_X \geq 0\ , \sum_{X\in[n]_m} p_X=1\ , \nonumber \\
	& && 1\geq t \geq 0\ . \nonumber 
	\end{align}
\end{lem}

\begin{rem}
	By the virtue  Lemma 2 from the main text, for $d=m=2$ the above semi-definite program computes the critical visibility for qubit POVMs i.e., for $d=2$ we have $t(\M)=t^{(2)}(\M)$. A qubit POVM $\M\in\P(2,n)$ is projective-simulable if and only if $t^{(2)}(\M)=1$.
\end{rem}

Using analogous methods and Lemma 3 form the main text, we obtain the following SDP programs for (i) checking whether a given qutrit POVM $\M$ is PM-simulable; and (ii) computation of the critical visibility $t(\M)$ for a qutrit POVM. 

\begin{lem}[SDP programme for PM-simulability for qutrits]\label{lem:qutritcharSDP}
	Let $\M\in\P(3,n)$ be a qutrit $n$-outcome measurement on $d$-dimensional Hilbert space. The question whether $\M$ is PM-simulable, $\M\in\SP(3,n)$, can be solved by the following SDP feasibility problem 
	\begin{align}
	& \text{Given} && \M = (M_1,\ldots,M_n)\nonumber \\
	& \text{Check if there exist} && \left\{\N_X\right\}_{X\in[n]_3}\ ,\ \left\{p_X\right\}_{X\in[n]_2}\ , \left\{\N_Y\right\}_{Y\in[n]_3}\ ,\ \left\{p_Y\right\}_{X\in[n]_3}\ , \nonumber \\
	& \text{Such that} && \M=\sum_{X\in[n]_2} \N_X + \sum_{Y\in[n]_3} \N_Y\ , \nonumber \\
	& && [\N_X]_i \geq 0\ ,\ [\N_Y]_i \geq 0\,\ \ i=1,\ldots,n \ , \nonumber \label{eq:convQUTRITout}\\
	& && [\N_X]_i = 0\ \ \text{for}\ i\notin{X} \ ,\ [\N_Y]_i = 0\ \ \text{for}\ i\notin{Y}\ , \ \nonumber \\
	& && \tr\left([\N_Y]_i\right) = p_Y\ \ \text{for}\ i\in{Y} \ , \\
	& && \sum_{i=1}^n [\N_X]_i= p_X \I\ , \sum_{i=1}^n [\N_Y]_i= p_Y \I\ \nonumber \\
	& && p_X \geq 0\ ,\ p_Y\geq 0\ ,\ \sum_{X\in[n]_2} p_X+\sum_{Y\in[n]_3} p_Y=1\ . \nonumber 
	\end{align}
\end{lem}

\begin{lem}[SDP for critical visibility for PM-simulability for qutrits simulability]\label{lem:QUTRITvisSDP}
	Let $\M\in\P(d,n)$ be $n$-outcome measurement on $d$-dimensional Hilbert space. Let $m\leq n$. The critical visibility $t(\M)$ can be computed via the following SPD programme
	\begin{align}
	& \text{Given} && \M = \left(M_1,\ldots,M_n\right)\nonumber \\
	& \text{Maximise} && t\ \nonumber \\
	& \text{Such that} && t\M+(1-t)\left(\tr(M_1)\frac{\I}{d},\ldots,\tr(M_n)\frac{\I}{d}\right)=\sum_{X\in[n]_2} \N_X + \sum_{Y\in[n]_3} \N_Y\ , \nonumber \\
	& \text{Where} && \left\{\N_X\right\}_{X\in[n]_3}\ ,\ \left\{p_X\right\}_{X\in[n]_2}\ , \left\{\N_Y\right\}_{Y\in[n]_3}\ ,\ \left\{p_Y\right\}_{X\in[n]_3}\ , \nonumber \\
	& && \M=\sum_{X\in[n]_2} \N_X + \sum_{Y\in[n]_3} \N_Y\ , \nonumber \\
	& && [\N_X]_i \geq 0\ ,\ [\N_Y]_i \geq 0\,\ \ i=1,\ldots,n \ , \nonumber \label{eq:convQUTRITout}\\
	& && [\N_X]_i = 0\ \ \text{for}\ i\notin{X} \ ,\ [\N_Y]_i = 0\ \ \text{for}\ i\notin{Y}\ , \ \nonumber \\
	& && \tr\left([\N_Y]_i\right) = p_Y\ \ \text{for}\ i\in{Y} \ , \\
	& && \sum_{i=1}^n [\N_X]_i= p_X \I\ , \sum_{i=1}^n [\N_Y]_i= p_Y \I\ \nonumber \\
	& && p_X \geq 0\ ,\ p_Y\geq 0\ ,\ \sum_{X\in[n]_2} p_X+\sum_{Y\in[n]_3} p_Y=1\ , \nonumber \\
	& && 1\geq t \geq 0\ . \nonumber 
	\end{align}
\end{lem}

\begin{rem}
	A qutrit POVM $\M\in\P(3,n)$ is projective-simulable if and only if $t(\M)=1$.
\end{rem}

\section{Details of the numerical computations} \label{app:Numerics}

\subsection*{External polytopes approximating $\P(2,4)$}

In what follows, we explicitly describe the construction of polytopes $\Delta\in\Herm(\C^2)^{\times 4}$ approximating $\P(4,2)$ \emph{from outside}. For $d=2$ the Bloch sphere representation facilitates the visualisation of the general scheme presented in the main text. Recall that operators $\{\I, \sigma_x, \sigma_y, \sigma_z\}$ form is a basis for the real space of Hermitian matrices $\Herm(\mathbb{C}^2)$. 
Therefore, any $M\in Herm(\mathbb{C}^2)$ can be written as
\begin{equation}\label{eq:papamQUB}
M = \alpha \I + x \sigma_x + y \sigma_y + z \sigma_z,
\end{equation}
where $\alpha, x, y, z \in \mathbb{R}$. Diagonalising $M$ we see that $M \geq 0$ if and only if
\begin{equation}
\alpha\geq0 \ \text{and}\ \sqrt{x^2+y^2+z^2}\leq \alpha\ . 
\end{equation}
The idea now is to relax the positive semi-definiteness condition $\M\geq0$ by requiring that $\tr(M\kb{\psi_j}{\psi_j}\geq0$ for some collection of pure states $\{\kb{\psi_j}{\psi_j}\}_{j=1}^N$. Every pure state has a representation, $\kb{\psi_j}{\psi_j}=(1/2)(\I-\vec{v}_j \cdot \vec{\sigma})$, where $\vec{v}_j$ is a vector from a unit sphere in $\R^3$. In the language of the parametrisation \eqref{eq:papamQUB} the conditions $\tr(M\kb{\psi_j}{\psi_j})\geq0$ are equivalent to
\begin{equation}\label{eq:facets}
(x,y,z)\cdot v_j \leq \alpha,\ i=1,...,N \ ,
\end{equation}
where ``$\cdot$'' denotes the standard inner product in $\R^3$. 

We define a polytope of quasi-POVMs $\Delta$ as a subset of $\Herm(\C^2)^{\times 4}$ consisting of vectors of Hermitian operators $\M=(M_1,\ldots,M_4)$ such that each ``effect'' satisfies inequalities \eqref{eq:facets}. Since $M\geq0$ implies \eqref{eq:facets} we conclude that $\P(2,4)\subset\Delta$.

We start writing a quasi-POVM as a 16-entry vector
\begin{equation}
\M \equiv (\alpha_1, x_1, y_1, z_1,..., \alpha_4, x_4, y_4, z_4)\in\R^{16},
\end{equation}
each four entries representing one quasi-effect. The initial description of $\Delta$ is given by $4 N$ inequalities defined by \eqref{eq:facets} applied plus the 8 ``global'' constraints
\begin{eqnarray}
&\alpha_i \geq 0, \ i=1,...,4\\
&\sum_i{\alpha_i} = 1\\
&\sum_i{x_i} = \sum_i{y_i} = \sum_i{z_i} = 0.
\end{eqnarray}

We now list a series of simplifications, with the goal of optimising the computation of lower bounds for $t(2)$.
\begin{itemize}
	\item The normalisation of the quasi-POVMs allows us to write $M_4 = \I- M_1 - M_2 -M_3$, and thus we can drop the 4 parameters describing $M_4$;
	\item Since $\P(2,4)$ and $\SP(2,4)$ are invariant under unitary conjugations, $\M\mapsto U.\M\defeq(U M_1 U^\dagger,\ldots,U M_4 U^\dagger)$, $t(\M)=t(U.\M)$. Hence we can restrict our construction to POVMs where (i) the Bloch vector of the first effect points in the direction of the $x$ axis; and (ii) the Bloch vector of the second effect lies in the $xy$ plane; 
	\item By doing this, we describe $M_1$ with only two parameters $\alpha_1, x_1$, and this description is exact: we can take $M_1$ to be a genuine effect (as opposed to a quasi-effect). Since extremal POVMs with $d^2$ outcomes have rank-one effects~\cite{DAriano2005}, we can set $\alpha_1 = x_1$ and dismiss one more parameter. With $M_2$ being described by three parameters, we pass from a 16-dimensional vector to a 8-dimensional vector representation of $\M$;
	\item Since $M_2$ lies in the plane, we consider a polygon $\hat{\Delta}$ approximating the circle \emph{from outside}, rather than a polyhedron approximation of the sphere.  Furthermore, we can always assume that this vector lies in the $y$-positive semi-plane, and take a polygon approximating the semi-circle. In what follows, we refer to the extremal points of of this polygon by $\{\vec{w}_i\}_{i=1}^{N'}$
	\item Since we have a candidate for the most robust POVM, $\M^{\mathrm{tetra}}$, we can adapt the construction of $\Delta$ to be tangent to a rotated $\M^{\mathrm{tetra}}$ to provide a better approximation around that point. We thus add to $\{\vec{v}\}_{j=1}^N$ the vector $(\cos\frac{2\pi}3,\sin\frac{2\pi}3,0)$, which corresponds to the second effect of $M^{\mathrm{tetra}}$, and add to vectors $\{\vec{v}_j\}_{j=1}^N$ the vertices $(-\frac12, -\frac{\sqrt3}{4}, \pm \frac34)$ corresponding to the third and fourth effects of $M^{\mathrm{tetra}}$.
\end{itemize}

Recall that our strategy to lower bound $t(2)$ was to construct a sequence of polytopes $\Delta\subset(\Herm(\mathbb{C}^2))^{\times 4}$ containing the set of four-outcome qubit POVMs $\P(2,4)$. Then we check using the SDP given in Lemma~\ref{lem:mOUTvisSDP} the minimum amount of depolarisation needed for the simulation of each extremal point of $\Delta$, and find the most robust one among the finitely many of them. The value $t_\Delta$ yielded this way is a lower bound for the minimum amount of depolarisation $t(2)$ for which any qubit POVM becomes simulable by projective measurements.

Setting $\{\vec{w}_i\}_{i=1}^{N'}$ to be the vertices of ``regular half-polygon'' of 100 sides and the remaining $\vec{v}_i$ to be the vertices of the Archimedean solid called truncated icosahedron (together with its dual), we obtain a polytope $\Delta$ with $\approx 850,000$ extremal points and $t_\Delta = 0.8152$, quite close to $t(\M^{\mathrm{tetra}})= \sqrt{2/3}\approx 0.8165$.

\subsection*{Potytopes approximating covariant measurements for qutrits}

Following the same reasoning applied to approximate $\P(2,4)$, we can approximate $\P(d,n)$ from the outside for any dimension $d$ and number of outcomes $n$.
One drawback is that already for dimension $d=3$ this becomes computationally expensive.
However, if we aim at the set $\P_{\mathrm{cov}}(3,9)$ of qutrit covariant POVMs regarding the discrete Heisenberg group $\mathbb{Z}^3 \times \mathbb{Z}^3$ \cite{Renes2003}  the task becomes easier, since for this class of POVMs from one effect we are able to derive the others.
Indeed, starting from a positive semi-definite seed $M\in\Herm(\C^3)$, having trace 1/3, we obtain the effects $M_1=M, M_2,...,M_9$ of a covariant POVM by conjugating $M$ by the unitaries
\begin{equation}
D_{jk} = \omega^{jk/2}\sum_{m=0}^2{\omega^{jk}\ket{k\oplus m}\bra{m}}\ ,
\end{equation}
where $\omega = \exp (2\pi i/d)$, $j,k=0,1,2$ and sum is modulus 3. We conjecture that among POVMs constructed in this way one  can find POVMs that are most robust regarding projective simulability. 

In order to construct a polytope of covariant quasi-POVMs, we  relax the positivity of the seed by ensuring $tr(M\ket{\phi_i}\bra{\phi_i})\geq 0$, for some finite set $\{\ket{\phi_i}\}_{i=1}^N$, and construct the corresponding quasi-effects.  
Since we need only $d^2-1=8$ real parameters to describe one Hermitian operator, this procedure in numerically trackable.

In dimension $d=3$ we do not have a particular candidate for the most robust POVM, as in the case of $\M^{\mathrm{tetra}}$ for qubits.
An heuristic search revealed that there are covariant POVMs with robustness $t(\M) \approx 0.7934$, leading to the bound
\begin{equation}
 t_{\mathrm{cov}}(3)\leq 0.7934\ ,
\end{equation} 
where $t_{\mathrm{cov}}(d)$ denotes the robustness of the most robust covariant POVM in $\P_{\mathrm{cov}}(d, d^2)$. Based on the approximation given by $N=180$ pure states chosen uniformly at random, we constructed a polytope of covariant quasi-POVMs with nearly 316.000 extremal points.
Optimizing over the critical visibilities over the extremal points we are able to find the lower bound to the critical vilibility
\begin{equation}
t_{\mathrm{cov}}(3)\geq 0.5559.
\end{equation} 

As mentioned in the main text, we do not have a characterization of $\P(d,d^2)$ that allows to decide whether a given POVM is projective-simulable or not for $d>3$. Nevertheless, we presented an SDP to check the simulation of a POVM in terms of $d$-outcome POVMs. It is worthy to mention that for every qutrit covariant POVM, the critical depolarization needed to make it projective-simulable matches the critical depolarization needed to make it 3-outcome-simulable.
We thus conjecture that for any covariant POVM, the best $d$-outcome simulators are already PMs.
If that is indeed the case, the strategy of constructing polytopes of covariant quasi-POVMs and optimizing over extremal points could be used to find arbitrarily sharp bounds on $t_{\mathrm{cov}}(d)$, for any dimension $d$, provided enough computational power.
Moreover, if the most robust POVMs are indeed covariant, then $t_{\mathrm{cov}}(d)=t(d)$ and this approach is general enough to convert a projective-LHV model into a LHV model for general POVMs in any dimension.

\subsection*{Projective decomposition of trace-one, three-outcome qutrit POVMs}

Here we present a constructive strategy to find a projective decomposition for trace-one qutrit measurements $\M\in\P_1(3,3)$. Our strategy is based on the proof of Lemma~\ref{lem:mainresult2TEH}.
To decrease the rank of each effect, we look for perturbations $\X=(X_1, X_2, X_3)\in\Herm(\H)^{\times 3}$ such that $\M\pm\X$ are still valid trace-one POVMs.
Trace preservation is equivalent to $\tr(X_i) = 0$ and the normalisation of POVMs requires that $\sum_i X_i = 0$. Recall that forf a POVM $\M$ we defined $r(\M)=(r_1,r_2,r_3)$, where $r_i=\mathrm{rank}(M_i)$ (permutations in the order of ranks correspond to relabellings of the original POVM).

We start assuming that the effects of $\M$ are full-rank, i.e., $r(\M)=(3,3,3).$ 
Let $\lbrace\ket{\psi_i}\rbrace_{i=1}^3$ be an arbitrary basis of $\C^3$ we define 
\begin{equation}\label{eq:pertX}
\PPP = \left(\kb{\psi_1}{\psi_1}, \kb{\psi_2}{\psi_2}, \kb{\psi_3}{\psi_3}\right) ,
\end{equation}
We now calculate 
\begin{equation}
t_\ast = \max \SET{t}{M_i - t \kb{\psi_i}{\psi_i} \geq 0} 
\end{equation}
noticing that this can be done via an SDP. This leads to a decomposition
\begin{equation}\label{bagoly}
\M = t_\ast \PPP  +(1-t_\ast)\tilde{\M} \ ,
\end{equation}
where $\tilde{\M}\defeq \frac{1}{1-t_\ast}  (\M -t_\ast\PPP)$ is a POVM belonging to $\P_1 \left(3,3\right)$. It is crucial that $\tilde{\M}$ has the property that ranks of its effects are smaller equal then the corresponding ranks of $\M$. Moreover, \emph{at least one effect} has rank smaller then the rank of the corresponding rank of the effect of $\M$. Having now the POVM $\tilde{M}$, we continue further decomposing, and, slightly abusing the notation, we keep referring to it as $\M$.

Let us now consider the cases when $r(\M)=(3,3,2)$ or $r(\M)=(3,2,2)$. The analogous reasoning as the one presented above for 
\begin{equation}
\tilde{\PPP}= \left(\kb{\phi_1}{\phi_1}, \kb{\phi_2}{\phi_2}, \kb{\phi_3}{\phi_3}\right)\ ,
\end{equation} 
where vectors $\ket{\phi_i}$ are chosen in such a way that $\ket{\phi_3}\in\mathrm{supp}(M_3)\cap\mathrm{supp}(M_3)$, $\ket{\phi_2}\in \mathrm{supp}(M_2)$ and $\ket{\phi_2}\perp\ket{\phi_3}$,  and $\ket{\phi_1}\perp\ket{\phi_2}$,  $\ket{\phi_1}\perp\ket{\phi_3}$.   In the course of this "rank reduction" process  \eqref{eq:pertX} provides a decomposition of $\M$ into a projective POVM and a POVM $\tilde{\M}\in \P_1 (3,3)$  has at least one effect with the rank smaller than $\M$.

 Notice that in the course of this ``rank reduction'' process, the case $r\M) = (2,1,3)$ will never occur because the second effect would be a projector and therefore orthogonal to the two others (remember that the trace of every effect of every POVM involved is one).

In the case $r(\M) = (2,2,2)$ we have to resort to the general method of perturbations. In the support of each of the three effects is spanned by a pair of eigenvectors. To find a traceless perturbation supported in the corresponding supports, we construct copies of the Pauli operators in each support. Since the space of traceless Hermitian operators acting on a space of dimension 3 has dimension 8, these 9 traceless operators cannot be all linearly independent. Exploring this fact, we construct $\X$ as it follows. Let $\ket{\psi_1^i}, \ket{\psi_2^i}$ be the eigenvectors associated to the non-null eigenvalues of $M_i,\ i=1,2,3$.
Consider now
\begin{eqnarray}
\sigma_x^i &=& \kb{\psi_1^i}{\psi_2^i} + \kb{\psi_2^i}{\psi_1^i}\ ,\\
\sigma_y^i &=& i(\kb{\psi_1^i}{\psi_2^i} - \kb{\psi_2^i}{\psi_1^i})\ ,\\ \sigma_z^i &=& \kb{\psi_1^i}{\psi_1^i} - \kb{\psi_2^i}{\psi_2^i}\ .
\end{eqnarray}
As mentioned above, without loss of generality we can write
\begin{equation}
\sigma_x^1 = \sum_{a,i}q_{a,i}\sigma_a^i\, 
\end{equation}
for some real coefficients $q_{a,i}$, with $q_{x,1} = 0$.
Thus the operators defined by
\begin{eqnarray}
X_1 &=& \sigma_x^1 - q_{y,1}\sigma_y^1 - q_{z,1}\sigma_z^1\\
X_i &=& -\sum_{a}q_{a,i}\sigma_a^i,\ i=2,3
\end{eqnarray}
are traceless, supported in the support of the corresponding $M_i$ and sum to zero. Having constructed the perturbation $\X$ We now calculate 
\begin{eqnarray}
t_+ = \max\{ t ; \M + t\X \in \P(3,3) \}\ , \\
t_- = \max\{ t ; \M - t\X \in \P(3,3) \}\ ,
\end{eqnarray}
noticing that this can be done via an SDP.
This leads to a decomposition
\begin{equation}\label{bagoly}
\M = \frac{t_-}{t_+ + t_-}\M_1 (\M + t_+\X) + \frac{t_+}{t_+ + t_-}(\M - t_-\X)\ ,
\end{equation}
where $\M_1=(\M + t_+\X)$ and $M_2=(\M - t_-\X)$ are POVMs belonging to $\P_1 (3,3)$. Again, $\M_1$,$\M_2$ both have the property that rank of their effects are smaller equal then the corresponding ranks of $\M$. Moreover, \emph{at least one their effects} has rank smaller then the rank of the corresponding rank of the effect of $\M$.

Finally, the $r(\M) = (1,2,2)$ is a trivial case, since it is equivalent to a two-outcome POVM in the subspace of $\H$ orthogonal to the support of the first effect.
The case $r(\M) = (1,1,2)$ will not occur (see reasoning justifying that $r(\M)=(3,2,1)$ would not occur). Finally, due to the conditions $\tr(M_i)=1$ the condition $r(\M)=(1,1,1)$ is equivalent to $\M$ being a projective measurement. Therefore, in the process of iterative of ``rank reduction'' described above we finally get a convex decomposition of the intitial POVM $\M$ into projective POVMs.

\end{document}